\documentclass[journal]{IEEEtran}
\usepackage{amsmath}
\usepackage{amssymb}
\usepackage{amsfonts}
\usepackage{graphicx}
\usepackage{epsfig}
\usepackage{subfigure}
\usepackage{psfrag}
\usepackage{color}

\begin{document}

\title{On Ergodic Sum Capacity of Fading Cognitive Multiple-Access and Broadcast Channels}

\author{Rui Zhang,~\IEEEmembership{Member,~IEEE,} Shuguang Cui,~\IEEEmembership{Member,~IEEE,}
and Ying-Chang Liang,~\IEEEmembership{Senior Member,~IEEE} 
\thanks{Manuscript received June 27, 2008; revised
April 13, 2009. This paper has been presented in part at Annual
Allerton Conference on Communication, Control and Computing,
Monticello, IL, USA, September 23-26, 2008.}
\thanks{R. Zhang and Y.-C. Liang are with the Institute for Infocomm Research, A*STAR, Singapore. (e-mails:
\{rzhang, ycliang\}@i2r.a-star.edu.sg)}
\thanks{S. Cui is with the Department of Electrical and Computer
Engineering, Texas A\&M University, Texas, USA. (e-mail:
cui@ece.tamu.edu)}}

\maketitle

\begin{abstract}
This paper studies the information-theoretic limits of a secondary
or cognitive radio (CR) network under spectrum sharing with an
existing primary radio network. In particular, the fading cognitive
multiple-access channel (C-MAC) is first studied, where multiple
secondary users transmit to the secondary base station (BS) under
both individual transmit-power constraints and a set of
interference-power constraints each applied at one of the primary
receivers. This paper considers the long-term (LT) or the short-term
(ST) transmit-power constraint over the fading states at each
secondary transmitter, combined with the LT or ST interference-power
constraint at each primary receiver. In each case, the optimal power
allocation scheme is derived for the secondary users to achieve the
ergodic sum capacity of the fading C-MAC, as well as the conditions
for the optimality of the dynamic time-division-multiple-access
(D-TDMA) scheme in the secondary network. The fading cognitive
broadcast channel (C-BC) that models the downlink transmission in
the secondary network is then studied under the LT/ST transmit-power
constraint at the secondary BS jointly with the LT/ST
interference-power constraint at each of the primary receivers. It
is shown that D-TDMA is indeed optimal for achieving the ergodic sum
capacity of the fading C-BC for all combinations of transmit-power
and interference-power constraints.
\end{abstract}


\IEEEpeerreviewmaketitle

\setlength{\baselineskip}{1.0\baselineskip}
\newtheorem{claim}{Claim}
\newtheorem{guess}{Conjecture}
\newtheorem{definition}{Definition}
\newtheorem{fact}{Fact}
\newtheorem{assumption}{Assumption}
\newtheorem{theorem}{\underline{Theorem}}[section]
\newtheorem{lemma}{\underline{Lemma}}[section]
\newtheorem{ctheorem}{Corrected Theorem}
\newtheorem{corollary}{\underline{Corollary}}[section]
\newtheorem{proposition}{Proposition}
\newtheorem{example}{\underline{Example}}[section]
\newtheorem{remark}{\underline{Remark}}[section]
\newtheorem{problem}{\underline{Problem}}[section]
\newtheorem{algorithm}{\underline{Algorithm}}[section]
\newcommand{\mv}[1]{\mbox{\boldmath{$ #1 $}}}

\section{Introduction}

\PARstart{C}ognitive radio (CR), since the name was coined by Mitola
in his seminal work \cite{Mitola00}, has drawn intensive attentions
from both academic (see, e.g., \cite{Goldsmith08} and references
therein) and industrial (see, e.g., \cite{802.22} and references
therein) communities; and to date, many interesting and important
results have been obtained. In CR networks, the secondary users or
CRs usually communicate over the same bandwidth originally allocated
to an existing primary radio network. In such a scenario, the CR
transmitters usually need to deal with a fundamental tradeoff
between maximizing the secondary network throughput and minimizing
the resulted performance degradation of the active primary
transmissions. One commonly known technique used by the secondary
users to protect the primary transmissions is {\it opportunistic
spectrum access} (OSA), originally outlined in \cite{Mitola00} and
later introduced by DARPA, whereby the secondary user decides to
transmit over a particular channel only when all primary
transmissions are detected to be off. For OSA, an enabling
technology is to detect the primary transmission on/off status, also
known as {\it spectrum sensing}, for which many algorithms have been
reported in the literature (see, e.g., \cite{YeLi08} and references
therein). However, in practical situations with a nonzero
misdetection probability for an active primary transmission, it is
usually impossible to completely avoid the performance degradation
of the primary transmission with the secondary user OSA.

Another approach different from OSA for a CR to maximize its
throughput and yet to provide sufficient protection to the primary
transmission is allowing the CR to access the channel even when the
primary transmissions are active, provided that the resultant
interference power, or the so-called {\it interference temperature}
(IT) \cite{Haykin05}, \cite{Gastpar07}, at each primary receiver is
limited below a predefined value. This spectrum sharing strategy is
also referred to as Spectrum Underlay \cite{Goldsmith08},
\cite{Zhao07} or Horizontal Spectrum Sharing \cite{Haykin05},
\cite{Tarokh06}. With this strategy, {\it dynamic resource
allocation} (DRA) becomes essential, whereby the transmit powers,
bit-rates, bandwidths, and antenna beams of the secondary
transmitters are dynamically allocated based upon the channel state
information (CSI) in the primary and secondary networks. A number of
papers have recently addressed the design of optimal DRA schemes to
achieve the point-to-point CR channel capacity under the IT
constraints at the primary receivers (see, e.g.,
\cite{Ghasemi07}-\cite{Zhang08d}). On the other hand, since the CR
network is in nature a multiuser communication environment, it will
be more relevant to consider DRA among multiple secondary users in a
CR network rather than that for the case of one point-to-point CR
channel. Deploying the interference-temperature constraint as a
practical means to protect the primary transmissions, the
conventional network models such as the multiple-access channel
(MAC), broadcast channel (BC), interference channel (IC), and relay
channel (RC) can all be considered for the secondary network,
resulting in various new cognitive network models and associated
problem formulations for DRA (see, e.g.,
\cite{ZhangLan08a}-\cite{Vu07}). It is also noted that there has
been study in the literature on the information-theoretic limits of
the CR channels by exploiting other types of ``cognitions''
available at the CR terminals different from the IT, such as the
knowledge of the primary user transmit messages at the CR
transmitter \cite{Tarokh06}, \cite{Viswanath06}, the distributed
detection results on the primary transmission status at the CR
transmitter and receiver \cite{Syed07a}, the ``soft'' sensing
results on the primary transmission \cite{Syed07b}, and the primary
transmission on-off statistics \cite{Meng08}.

In this paper, we focus on the single-input single-output (SISO) or
single-antenna fading cognitive MAC (C-MAC) and cognitive BC (C-BC)
for the secondary network, where $K$ secondary users communicate
with the base station (BS) of the secondary network in the presence
of $M$ primary receivers. It is assumed that the BS has the perfect
CSI on the channels between the BS and all the secondary users, as
well as the channels from the BS and each secondary user to all the
primary receivers.\footnote{In practice, CSI on the channels between
the secondary users and their BS can be obtained by the classic
channel training, estimation, and feedback mechanisms, while CSI on
the channels between the secondary BS/users and the primary
receivers can be obtained by the secondary BS/users via, e.g.,
estimating the received signal power from each primary terminal when
it transmits, under the assumptions of pre-knowledge on the primary
transmit power levels and the channel reciprocity.} Thereby, the BS
can implement a centralized dynamic power and rate allocation scheme
in the secondary network so as to optimize its performance and yet
maintain the interference power levels at all the primary receivers
below the prescribed thresholds. An information-theoretic approach
is taken in this paper to characterize the maximum sum-rate of
secondary users averaged over the channel fading states, termed as
{\it ergodic sum capacity}, for both the fading C-MAC and C-BC. The
ergodic sum capacity can be a relevant measure for the maximum
achievable throughput of the secondary network when the data traffic
has a sufficiently-large delay tolerance. As usual (see, e.g.,
\cite{Caire99}), we consider both the long-term (LT) transmit-power
constraint (TPC) that regulates the {\it average} transmit power
across all the fading states at the BS or each of the secondary
user, as well as the short-term (ST) TPC that is more restrictive
than the LT-TPC by limiting the {\it instantaneous} transmit power
at each fading state to be below a certain threshold. Similarly, we
also consider both the LT interference-power constraint (IPC) that
regulates the resultant average interference power over fading at
each primary receiver, and the ST-IPC that imposes a more strict
instantaneous limit on the resultant interference power at each
fading state. The major problem to be addressed in this paper is
then to characterize the ergodic sum capacity of the secondary
network under different combinations of LT-/ST-TPC and LT-/ST-IPC.
Apparently, such a problem setup is unique for the fading CR
networks. Moreover, we are interested in investigating the
conditions over each case for the optimality of the dynamic
time-division-multiple-access (D-TDMA) scheme in the secondary
network, i.e., when it is optimal to schedule a single secondary
user at each fading state for transmission to achieve the ergodic
sum capacity. These optimality conditions for D-TDMA are important
to know as when they are satisfied, the single-user decoding and
encoding at the secondary BS becomes optimal for the C-MAC and C-BC,
respectively. This can lead to a significant complexity reduction
compared with the cases where these conditions are not satisfied
such that the BS requires more complex multiuser decoding and
encoding techniques to achieve the ergodic sum capacity.

Information-theoretic studies can be found for the deterministic (no
fading) SISO-MAC and SISO-BC in, e.g., \cite{Coverbook}, and for the
fading (parallel) SISO-MAC and SISO-BC in, e.g.,
\cite{Verdu93}-\cite{Tse98b} and \cite{HH75}-\cite{Li01a},
respectively. In addition, D-TDMA has been shown as the optimal
transmission scheme to achieve the ergodic sum capacity of the
fading SISO-MAC under the LT-TPC at each transmitter \cite{Tse98a},
\cite{Knopp95}. Thanks to the duality result on the capacity regions
of the Gaussian MAC and BC \cite{Goldsmith04}, the optimality of
D-TDMA is also provable for the fading SISO-BC to achieve the
ergodic sum capacity. However, to our best knowledge,
characterizations of the ergodic sum capacities as well as the
optimality conditions for D-TDMA over the fading C-MAC and C-BC
under various mixed transmit-power and interference-power
constraints have not been addressed yet in the literature. In this
paper, we will provide the solutions to these problems. The main
results of this paper are summarized below for a brief overview:
\begin{itemize}
\item For the fading cognitive SISO-MAC, we show that D-TDMA is
optimal for achieving the ergodic sum capacity when the LT-TPC is
applied jointly with the LT-IPC. This result is an extension of that
obtained earlier in \cite{Knopp95} for the traditional fading
SISO-MAC without the LT-IPC. For the other three cases of mixed
power constraints, i.e., LT-TPC with ST-IPC, ST-TPC with LT-IPC, and
ST-TPC with ST-IPC, we show that although D-TDMA is in general a
suboptimal scheme and thus does not achieve the ergodic sum
capacity, it can be optimal under some special conditions. We
formally derive these conditions from the Karush-Kuhn-Tucker (KKT)
conditions \cite{Boydbook} associated with the capacity maximization
problems. In particular, for the case of LT-TPC with ST-IPC, we show
that the optimal number of secondary users that transmit at the same
time should be no greater than $M+1$. Therefore, for small values of
$M$, e.g., $M=1$ corresponding to a single primary receiver, D-TDMA
is close to being optimal. Furthermore, for all cases considered, we
derive the optimal transmit power-control policy for the secondary
users to achieve the ergodic sum capacity. For the two cases of
LT-TPC with LT-IPC and ST-TPC with LT-IPC, we provide the
closed-form solutions for the optimal power allocation at each
fading state. Particularly, in the case of ST-TPC with LT-IPC, we
show that for the active secondary users at one particular fading
state, there is at most one user that transmits with power lower
than its ST power constraint, while all the other active users
transmit with their maximum powers.

\item For the fading cognitive SISO-BC, we show that for all considered cases
of mixed power constraints, D-TDMA is optimal for achieving the
ergodic sum capacity. The optimal transmit power allocations at the
BS in these cases have closed-form solutions, which resemble the
single-user ``water-filling (WF)'' solutions for the well-known
fading (parallel) Gaussian channels \cite{Coverbook},
\cite{Goldsmith97}.
\end{itemize}

The rest of this paper is organized as follows. Section
\ref{sec:system model} provides the system model for the fading
C-MAC and C-BC. Section \ref{sec:MAC} and Section \ref{sec:BC} then
present the results on the ergodic sum-capacity, the associated
optimal power-control policy, and the optimality conditions for
D-TDMA, for the fading C-MAC and C-BC, respectively, under different
mixed LT/ST transmit-power and interference-power constraints.
Section \ref{sec:numerical results} provides the numerical results
on the ergodic sum capacities of the fading C-MAC and C-BC under
different mixed power constraints, the capacities with vs. without
the TDMA constraint, and those with vs. without the optimal power
control, and draws some insightful observations pertinent to the
optimal DRA in CR networks. Finally, Section \ref{sec:conclusions}
concludes this paper.

\section{System Model} \label{sec:system model}

\begin{figure}
\psfrag{a}{SU-1} \psfrag{b}{SU-2} \psfrag{c}{SU-K} \psfrag{d}{PR-1}
\psfrag{e}{PR-2} \psfrag{f}{PR-M} \psfrag{g}{BS} \psfrag{h}{$h_1$}
\psfrag{i}{$h_2$} \psfrag{j}{$h_K$}
\psfrag{k}{$g_{11}$}\psfrag{l}{$g_{12}$}\psfrag{m}{$g_{K2}$}\psfrag{n}{$g_{KM}$}
\begin{center}
\scalebox{0.9}{\includegraphics*[52pt,542pt][331pt,771pt]{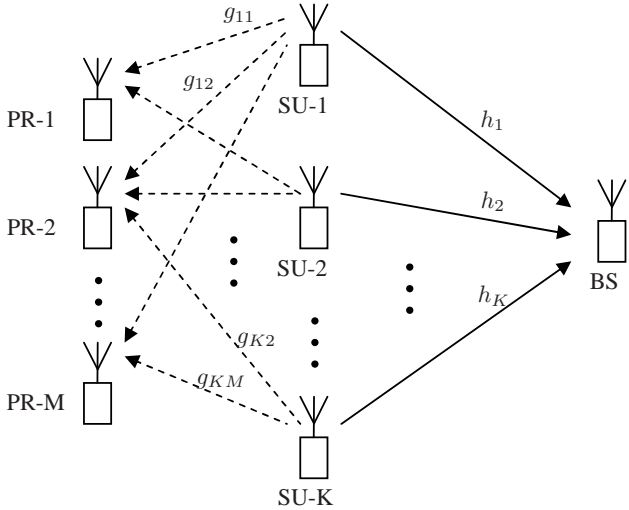}}
\end{center}\vspace{-0.2in}
\caption{The cognitive SISO-MAC where $K$ SUs transmit to the secondary BS while possibly interfering with each of $M$ PRs.} \label{fig:C-MAC
model}\vspace{-0.1in}
\end{figure}

Consider a fading C-MAC as shown in Fig. \ref{fig:C-MAC model},
where $K$ CRs or secondary users (SUs) transmit to the secondary BS
by sharing the same narrow band with $M$ primary receivers (PRs),
and all terminals are assumed to be equipped with a single antenna
each. A {\it block-fading} (BF) channel model is assumed for all the
channels involved. Furthermore, since this paper considers coherent
communications, only the fading channel power gains (amplitude
squares) are of interest. During each transmission block, the power
gain of the fading channel from the $k$-th SU to the secondary BS is
denoted by $h_k$, while that of the fading channel from the $k$-th
SU to the $m$-th PR is denoted by $g_{km}$, $k=1,\ldots,K,
m=1,\ldots,M$. These channel power gains are assumed to be drawn
from a vector random process, which we assume to be ergodic over
transmission blocks and have a continuous, differentiable joint
cumulative distribution function (cdf), denoted by $F(\mv{\alpha})$,
where $\mv{\alpha}\triangleq[h_1\cdots h_K, g_{11}\cdots g_{1M},
g_{21}\cdots g_{2M},\ldots, g_{K1}\cdots g_{KM}]$ denotes the power
gain vector for all the channels of interest. We further assume that
$h_k$'s and $g_{km}$'s are independent. In addition, it is assumed
that the additive noises (including any additional interferences
from the outside of the secondary network, e.g., the primary
transmitters) at the secondary BS are independent circular symmetric
complex Gaussian (CSCG) random variables, each having zero mean and
unit variance, denoted as $\mathcal{CN}(0,1)$. Since in this paper
we are interested in the information-theoretic limits of the C-MAC,
it is assumed that the optimal Gaussian codebook is used by each SU
transmitter.

It is assumed that the secondary BS knows {\it a priori} the channel
distribution information $F(\mv{\alpha})$ and furthermore the
channel realization $\mv{\alpha}$ at each transmission block.
Thereby, the secondary BS is able to schedule transmissions of SUs
and allocate their transmit power levels and rate values at each
transmission block, so as to optimize the performance of the
secondary network and yet provide a necessary protection to each of
the PRs. We denote the transmit power-control policy for SUs as
$\mathcal{P}_{\rm MAC}$, which specifies a mapping from the fading
channel realization $\mv{\alpha}$ to
$\mv{p}(\mv{\alpha})\triangleq[p_1(\mv{\alpha}),\ldots,p_K(\mv{\alpha})]$,
where $p_k(\mv{\alpha})$ denotes the transmit power assigned to the
$k$-th SU. The long-term (LT) transmit-power constraint (TPC) for
the $k$-th SU, $k=1,\ldots,K$, can then be described as
\begin{equation}\label{eq:LT TPC MAC}
\mathbb{E}\left[p_k(\mv{\alpha})\right]\leq P^{\rm LT}_k
\end{equation}
where the expectation is taken over $\mv{\alpha}$ with respect to
(w.r.t.) its cdf, $F(\mv{\alpha})$, and the short-term (ST)
transmit-power constraint (TPC) for the $k$-th SU is given as
\begin{equation}\label{eq:ST TPC MAC}
p_k(\mv{\alpha})\leq P^{\rm ST}_k, \ \forall \mv{\alpha}.
\end{equation}
Similarly, we consider both the LT and ST interference-power
constraints (IPCs) at the $m$-th PR, $m=1,\ldots,M$, described as
\begin{equation}\label{eq:LT IPC MAC}
\mathbb{E}\left[\sum_{k=1}^Kg_{km}p_k(\mv{\alpha})\right]\leq \Gamma^{\rm LT}_m
\end{equation}

\begin{equation}\label{eq:ST IPC MAC}
\sum_{k=1}^Kg_{km}p_k(\mv{\alpha})\leq \Gamma^{\rm ST}_m, \ \forall
\mv{\alpha},
\end{equation}
respectively. For a given $\mathcal{P}_{\rm MAC}$, the maximum
achievable sum-rate (in nats/complex dimension) of SUs averaged over
all the fading states can be expressed as (see, e.g.,
\cite{Shamai97})
\begin{equation}\label{eq:sum rate MAC}
R_{\rm MAC}(\mathcal{P}_{\rm MAC})=\mathbb{E}\left[
\log\left(1+\sum_{k=1}^Kh_kp_k(\mv{\alpha})\right)\right].
\end{equation}
The ergodic sum capacity of the fading C-MAC can then be defined as
\begin{equation}\label{eq:sum capacity MAC}
C_{\rm MAC}=\max_{\mathcal{P}_{\rm MAC}\in \mathcal{F}} R_{\rm MAC}(\mathcal{P}_{\rm MAC})
\end{equation}
where $\mathcal{F}$ is the feasible set specified by a particular
combination of the LT-TPC, ST-TPC, LT-IPC and ST-IPC. Note that all
of these power constraints are affine and thus specify convex sets
of $p_k(\mv{\alpha})$'s, so does any of their arbitrary
combinations. Therefore, the capacity maximization in (\ref{eq:sum
capacity MAC}) is in general a {\it convex} optimization problem,
and thus efficient numerical algorithms are available to obtain its
solutions. In this paper, we consider $\mathcal{F}$ to be generated
by one of the following four possible combinations of power
constraints, which are LT-TPC with LT-IPC, LT-TPC with ST-IPC,
ST-TPC with LT-IPC, and ST-TPC with ST-IPC, for the purpose of
exposition.

\begin{figure}
\psfrag{a}{SU-1} \psfrag{b}{SU-2} \psfrag{c}{SU-K} \psfrag{d}{PR-1}
\psfrag{e}{PR-2} \psfrag{f}{PR-M} \psfrag{g}{BS} \psfrag{h}{$h_1$}
\psfrag{i}{$h_2$} \psfrag{j}{$h_K$}
\psfrag{k}{$f_1$}\psfrag{l}{$f_2$}\psfrag{m}{$f_M$}
\begin{center}
\scalebox{0.9}{\includegraphics*[52pt,542pt][331pt,771pt]{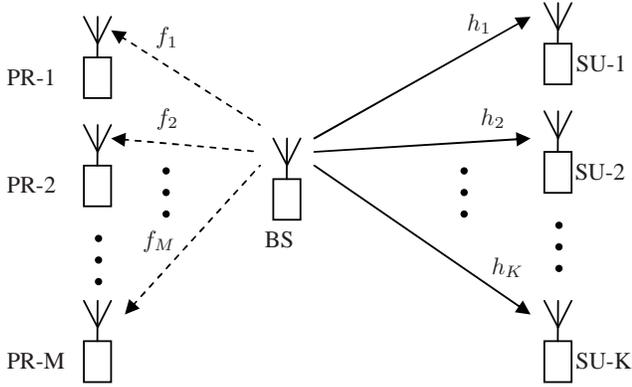}}
\end{center}\vspace{-0.4in}
\caption{The cognitive SISO-BC where the secondary BS transmits to $K$ SUs while possibly interfering with each of $M$ PRs.} \label{fig:C-BC
model}\vspace{-0.1in}
\end{figure}

Next, we consider the SISO fading C-BC as shown in Fig.
\ref{fig:C-BC model}, where the secondary BS transmits to $K$ SUs
while possibly interfering with each of the $M$ PRs. Without loss of
generality, we use the same notation, $h_k$, to denote the channel
power gain from the BS to the $k$-th SU, $k=1,\ldots,K$, as for the
C-MAC. The interference channel power gains from the BS to PRs are
denoted as $f_m$, $m=1,\ldots,M$, which are assumed to be mutually
independent and also independent of $h_k$'s. Similar to the C-MAC
case, let $\mv{\beta}\triangleq[h_1\cdots h_K, f_1\cdots f_M]$
denote the power gain vector for all the channels involved in the
C-BC, which we assume to be drawn from an ergodic vector random
process with a continuous, differentiable joint cdf, denoted by
$G(\mv{\beta})$. It is assumed that the additive noises at all SU
receivers are independent CSCG random variables each distributed as
$\mathcal{CN}(0,1)$; and the optimal Gaussian codebook is used by
the transmitter of the BS. With the available channel distribution
information $G(\mv{\beta})$ as well as the CSI on $h_k$'s and
$f_m$'s at each transmission block, the secondary BS designs its
downlink transmissions to the SUs by dynamically allocating its
transmit power levels and rate values. Let $\mathcal{P}_{\rm BC}$
denote the transmit power-control policy for the secondary BS, which
specifies a mapping from the fading channel realization $\mv{\beta}$
to its transmit power $q(\mv{\beta})$. Similarly as for C-MAC, we
define the LT-TPC and ST-TPC for the secondary BS as
\begin{equation}\label{eq:LT TPC BC}
\mathbb{E}\left[q(\mv{\beta})\right]\leq Q^{\rm LT}
\end{equation}
where the expectation is taken over $\mv{\beta}$ w.r.t. its cdf,
$G(\mv{\beta})$, and
\begin{equation}\label{eq:ST TPC BC}
q(\mv{\beta}) \leq Q^{\rm ST}, \ \forall \mv{\beta},
\end{equation}
respectively; and the LT-IPC and ST-IPC at the $m$-th PR,
$m=1,\ldots,M$, as
\begin{equation}\label{eq:LT IPC BC}
\mathbb{E}\left[f_{m}q(\mv{\beta})\right]\leq \Gamma^{\rm LT}_m
\end{equation}
and
\begin{equation}\label{eq:ST IPC BC}
f_{m}q(\mv{\beta})\leq \Gamma^{\rm ST}_m, \ \forall \mv{\beta},
\end{equation}
respectively.

Now, consider an auxiliary SISO fading C-MAC for the SISO fading
C-BC of interest, where $h_{k}$'s remain the same as in the C-BC
while $g_{km}=f_m, \forall k\in\{1,\ldots,K\}, m\in\{1,\ldots,M\}$.
Thus, the channel realization $\mv{\alpha}$ in this auxiliary C-MAC
can be concisely represented by $\mv{\beta}$ in the C-BC. By
applying the MAC-BC duality result \cite{Goldsmith04} at each fading
state, for a given $q(\mv{\beta})$, the maximum sum-rate of the C-BC
can be obtained from its auxiliary C-MAC as
\begin{equation}
\max_{\sum_{k=1}^Kp_k(\mv{\beta})=q(\mv{\beta})}
\log\left(1+\sum_{k=1}^Kh_kp_k(\mv{\beta})\right).
\end{equation}
Therefore, the ergodic sum capacity of the fading C-BC can be
equivalently obtained from its auxiliary fading C-MAC as
\begin{equation}\label{eq:sum capacity BC new}
C_{\rm BC}=\max_{\mathcal{P}_{\rm MAC}\in \mathcal{D}} R_{\rm
MAC}(\mathcal{P}_{\rm MAC}).
\end{equation}
where $\mathcal{D}$ is specified by a particular combination of
(\ref{eq:LT TPC BC})-(\ref{eq:ST IPC BC}), with $q(\mv{\beta})$
being replaced by $\sum_{k=1}^Kp_k(\mv{\beta})$. Note that we can
obtain the optimal power-control policy $\mathcal{P}_{\rm BC}$ to
achieve the ergodic sum capacity of the C-BC from the corresponding
optimal $\mathcal{P}_{\rm MAC}$ by solving the maximization problem
in (\ref{eq:sum capacity BC new}). Similarly as for $C_{\rm MAC}$ in
(\ref{eq:sum capacity MAC}), it can be shown that the optimization
problem for obtaining $C_{\rm BC}$ in (\ref{eq:sum capacity BC new})
is convex.

\section{Ergodic Sum Capacity for Fading Cognitive MAC}
\label{sec:MAC}

In this section, we consider the SISO fading C-MAC under different
mixed transmit-power and interference-power constraints. For each
case, we derive the optimal power-control policy for achieving the
ergodic sum capacity, as well as the conditions for the optimality
of D-TDMA.

\subsection{Long-Term Transmit-Power and Interference-Power Constraints}

\begin{figure*}
\begin{align}\label{eq:Lagrangian LT LT}
\mathcal{L}(\{p_k(\mv{\alpha})\},\{\lambda_k\},\{\mu_m\})=\mathbb{E}\left[\log(1+\sum_{k=1}^Kh_kp_k(\mv{\alpha}))\right]-\sum_{k=1}^K\lambda_k\{
\mathbb{E}[p_k(\mv{\alpha})]- P^{\rm LT}_k\} -\sum_{m=1}^M\mu_m\left\{\mathbb{E}\left[\sum_{k=1}^Kg_{km}p_k(\mv{\alpha})\right]- \Gamma^{\rm
LT}_m\right\}
\end{align}
\end{figure*}

From (\ref{eq:sum rate MAC}) and (\ref{eq:sum capacity MAC}), the
ergodic sum capacity under the LT-TPC and the LT-IPC can be obtained
by solving the following optimization problem:
\begin{problem}\label{prob:LT LT MAC}
\begin{eqnarray*}
\mathop{\mathtt{Maximize \ (Max.)}}_{\{p_k(\mv{\alpha})\}}&&
\mathbb{E}\left[
\log\left(1+\sum_{k=1}^Kh_kp_k(\mv{\alpha})\right)\right]
\\ \mathtt {subject \ to \ (s.t.)} && (\ref{eq:LT TPC MAC}), (\ref{eq:LT IPC
MAC}).
\end{eqnarray*}
\end{problem}
The proposed solution to the above problem is based on the Lagrange duality method. First, we write the Lagrangian of this problem as in
(\ref{eq:Lagrangian LT LT}) (shown on the next page), where $\lambda_k$ and $\mu_m$ are the nonnegative dual variables associated with each
corresponding power constraint in (\ref{eq:LT TPC MAC}) and (\ref{eq:LT IPC MAC}), respectively, $k=1,\ldots,K$, $m=1,\ldots,M$. Then, the
Lagrange dual function, $g(\{\lambda_k\},\{\mu_m\})$, is defined as
\begin{eqnarray}\label{eq:Lagrange dual LT LT}
\max_{\{p_k(\mv{\alpha})\}: p_k(\mv{\alpha})\geq 0, \forall k, \mv{\alpha} }\mathcal{L}(\{p_k(\mv{\alpha})\},\{\lambda_k\},\{\mu_m\}).
\end{eqnarray}
The dual function serves as an upper bound on the optimal value of
the original (primal) problem, denoted by $r^{*}$, i.e., $r^*\leq
g(\{\lambda_k\},\{\mu_m\})$ for any nonnegative $\lambda_k$'s and
$\mu_m$'s. The dual problem is then defined as
\begin{equation}\label{eq:dual problem LT LT}
\min_{\{\lambda_k\},\{\mu_m\}:\lambda_k\geq 0, \mu_m\geq 0, \forall
k, m} g(\{\lambda_k\},\{\mu_m\}).
\end{equation}
Let the optimal value of the dual problem be denoted by $d^*$, which
is achievable by the optimal dual solutions $\{\lambda_k^*\}$ and
$\{\mu_m^*\}$, i.e., $d^*= g(\{\lambda_k^*\},\{\mu_m^*\})$. For a
convex optimization problem with a strictly feasible point as in our
problem, the Slater's condition \cite{Boydbook} is satisfied and
thus the duality gap, $r^* - d^* \leq 0$, is indeed zero. This
result ensures that Problem \ref{prob:LT LT MAC} can be equivalently
solved from its dual problem, i.e., by first maximizing its
Lagrangian to obtain the dual function for some given dual
variables, and then minimizing the dual function over the dual
variables.

Consider first the problem for obtaining
$g(\{\lambda_k\},\{\mu_m\})$ with some given $\lambda_k$'s and
$\mu_m$'s. It is interesting to observe that this dual function can
also be written as
\begin{eqnarray}\label{eq:dual function rewrite LT LT}
g(\{\lambda_k\},\{\mu_m\})= \mathbb{E}\left[g'(\mv{\alpha})\right] + \sum_{k=1}^K\lambda_kP^{\rm LT}_k + \sum_{m=1}^M \mu_m\Gamma^{\rm LT}_m
\end{eqnarray}
where
\begin{align}\label{eq:dual function per state LT LT}
g'(\mv{\alpha})=&\max_{\{p_k(\mv{\alpha})\}: p_k(\mv{\alpha})\geq 0, \forall k}\log\left(1+\sum_{k=1}^Kh_kp_k(\mv{\alpha})\right) \nonumber \\
& - \sum_{k=1}^K\lambda_kp_k(\mv{\alpha}) -\sum_{m=1}^M\mu_m\sum_{k=1}^Kg_{km}p_k(\mv{\alpha}).
\end{align}
Thus, the dual function can be obtained via solving for
sub-dual-function $g'(\mv{\alpha})$'s, each for one fading state
with channel realization, $\mv{\alpha}$. Notice that the
maximization problems in (\ref{eq:dual function per state LT LT})
with different $\mv{\alpha}$'s all have the same structure and thus
can be solved using the same computational routine. For conciseness,
we drop the $\mv{\alpha}$ in $p_k($\mv{\alpha}$)$'s for the
maximization problem at each fading state and express it as
\begin{problem}\label{prob:LT LT MAC per state}
\begin{align}
\mathop{\mathtt{Max.}}_{\{p_k\}} & ~~\log\left(1+\sum_{k=1}^Kh_kp_k\right)-  \sum_{k=1}^K\lambda_kp_k -\sum_{m=1}^M\mu_m\sum_{k=1}^Kg_{km}p_k
\\ \mathtt {s.t.} &~~ p_k\geq 0, \ \forall k.
\end{align}
\end{problem}
This problem is convex since its objective function is concave and
its constraints are all linear. By introducing nonnegative dual
variables $\delta_k, k=1,\ldots, K$, for the corresponding
constraints on the nonnegativity of $p_k$'s, we can write the
following KKT conditions \cite{Boydbook} that need to be satisfied
by the optimal primal and dual solutions of Problem \ref{prob:LT LT
MAC per state}, denoted as $\{p_k^*\}$ and $\{\delta_k^*\}$,
respectively.
\begin{align}
\frac{h_k}{1+\sum_{l=1}^Kh_lp_l^*}-\lambda_k-\sum_{m=1}^M\mu_mg_{km}+\delta_k^*&=0,
\forall k \label{eq:KKT1 LT LT} \\
\delta_k^*p_k^*&=0, \forall k \label{eq:KKT2 LT LT}
\end{align}
with $p_k^*\geq 0$ and $\delta_k^*\geq 0, \forall k$. The following
lemma can then be obtained from these KKT optimality conditions:
\begin{lemma}\label{lemma:TDMA opt LT LT}
The optimal solution of Problem \ref{prob:LT LT MAC per state} has
at most one user indexed by $i$, $i\in\{1,\ldots,K\}$, with $p_i^*>
0$, i.e., the solution follows a D-TDMA structure.
\end{lemma}
\begin{proof}
Please refer to Appendix \ref{appendix:proof TDMA opt LT LT}.
\end{proof}
Given Lemma \ref{lemma:TDMA opt LT LT}, the remaining tasks for
solving Problem \ref{prob:LT LT MAC per state} are to find the user
that transmits at each fading state as well as the optimal transmit
power, which are given by the following lemma:
\begin{lemma}\label{lemma:opt power user LT LT}
In the optimal solution of Problem \ref{prob:LT LT MAC per state},
let $i$ denote the user that has $p_i^*> 0$, and $j$ be any of the
other users that has $p_j^*=0$, $i,j\in\{1,\ldots,K\}$. Then user
$i$ must satisfy
\begin{equation}\label{eq:user ordering LT LT}
\frac{h_i}{\lambda_i+\sum_{m=1}^M\mu_m
g_{im}}\geq\frac{h_j}{\lambda_j+\sum_{m=1}^M\mu_m g_{jm}}, \ \forall
j\neq i.
\end{equation}
The optimal power allocation of user $i$ is
\begin{equation}\label{eq:opt power LT LT}
p_i^*=\left( \frac{1}{\lambda_i+\sum_{m=1}^M\mu_m g_{im}}-\frac{1}{h_i}\right)^+
\end{equation}
where $(x)^+=\max(0,x)$.
\end{lemma}
\begin{proof}
Please refer to Appendix \ref{appendix:proof opt user power LT LT}.
\end{proof}

Solutions of Problem \ref{prob:LT LT MAC per state} across all the
fading states are basically an optimal mapping between an arbitrary
channel realization and the transmit power allocation for any given
$\lambda_k$'s and $\mu_m$'s, which can then be used to obtain the
dual function $g(\{\lambda_k\},\{\mu_m\})$. Next, the dual function
needs to be minimized over $\lambda_k$'s and $\mu_m$'s to obtain the
optimal dual solutions $\lambda_k^*$'s and $\mu_m^*$'s with which
the duality gap is zero. One method to iteratively update
$\lambda_k$'s and $\mu_m$'s toward their optimal values is the
ellipsoid method \cite{BGT81}, of which we omit the details here for
brevity.

Lemma \ref{lemma:TDMA opt LT LT} suggests that at each fading state,
at most one SU can transmit, i.e., D-TDMA is optimal. Since this
result holds for any given $\lambda_k$'s and $\mu_m$'s, it must be
true for the optimal dual solutions $\lambda_k^*$'s and $\mu_m^*$'s
under which the optimal value of the original problem or the ergodic
sum capacity is achieved. Therefore, we have the following theorem:

\begin{theorem}\label{theorem:TDMA opt LT LT}
D-TDMA is optimal across all the fading states for achieving the
ergodic sum capacity of the fading C-MAC under the LT-TPC jointly
with the LT-IPC. The optimal rules to select the SU for transmission
at a particular fading state and to determine its transmit power are
given by Lemma \ref{lemma:opt power user LT LT} with all
$\lambda_k$'s and $\mu_m$'s replaced by their optimal dual solutions
for Problem \ref{prob:LT LT MAC}.
\end{theorem}

\begin{remark}
Notice that if the LT-IPC given by (\ref{eq:LT IPC MAC}) is not
present in Problem \ref{prob:LT LT MAC}, or equivalently, the LT-IPC
values $\Gamma_m^{\rm LT}$'s are sufficiently large such that these
constraints are inactive with the optimal power solutions of Problem
\ref{prob:LT LT MAC}, it is then easy to verify from its KKT
conditions that the optimal dual solutions for all $\mu_m$'s must be
equal to zero. From (\ref{eq:user ordering LT LT}), it then follows
that only user $i$ with the largest $\frac{h_i}{\lambda_i}$ among
all the users can probably transmit at a given fading state. This
result is consistent with that obtained earlier in \cite{Knopp95}
for the traditional fading SISO-MAC without the LT-IPC. However,
under the additional LT-IPC, from (\ref{eq:user ordering LT LT}) and
(\ref{eq:opt power LT LT}) it is observed that the selected SU for
transmission and its transmit power depend on the interference-power
``prices'' $\mu_m$'s for different PRs and the instantaneous
interference channel power gains $g_{km}$'s.
\end{remark}

\subsection{Long-Term Transmit-Power and Short-Term
Interference-Power Constraints}

The ergodic sum capacity under the LT-TPC but with the ST-IPC can be
obtained as the optimal value of the following problem:
\begin{problem}\label{prob:LT ST MAC}
\begin{eqnarray*}
\mathop{\mathtt{Max.}}_{\{p_k(\mv{\alpha})\}}&& \mathbb{E}\left[
\log\left(1+\sum_{k=1}^Kh_kp_k(\mv{\alpha})\right)\right]
\\ \mathtt {s.t.} && (\ref{eq:LT TPC MAC}), (\ref{eq:ST IPC
MAC}).
\end{eqnarray*}
\end{problem}
Similar to Problem \ref{prob:LT LT MAC}, we apply the Lagrange
duality method to solve the above problem. However, different from
Problem \ref{prob:LT LT MAC} that has both the long-term
transmit-power and interference-power constraints, it is noted that
in Problem \ref{prob:LT ST MAC}, only the transmit-power constraints
are long-term while the interference-power constraints are
short-term. Therefore, the dual variables associated with the
long-term constraints should be introduced first, in order to
decompose the problem into individual subproblems over different
fading states, to each of which the corresponding short-term
constraints can then be applied. Let $\lambda_k$ be the nonnegative
dual variable associated with the corresponding LT-TPC in
(\ref{eq:LT TPC MAC}), $k=1,\ldots,K$. The Lagrangian of this
problem can then be written as
\begin{align}\label{eq:Lagrangian LT ST}
\mathcal{L}(\{p_k(\mv{\alpha})\},\{\lambda_k\})=&~\mathbb{E}\left[ \log\left(1+\sum_{k=1}^Kh_kp_k(\mv{\alpha})\right)\right] \nonumber
\\& ~ -\sum_{k=1}^K\lambda_k\left\{ \mathbb{E}\left[p_k(\mv{\alpha})\right]- P^{\rm LT}_k\right\}.
\end{align}
Let $\mathcal{A}$ denote the set of $\{p_k(\mv{\alpha})\}$ specified
by the remaining ST-IPC in (\ref{eq:ST IPC MAC}). The Lagrange dual
function is then expressed as
\begin{eqnarray}\label{eq:Lagrange dual LT ST}
g(\{\lambda_k\})=\max_{\{p_k(\mv{\alpha})\}\in\mathcal{A}}\mathcal{L}(\{p_k(\mv{\alpha})\},\{\lambda_k\}).
\end{eqnarray}
The dual problem is accordingly defined as $\min_{\lambda_k\geq 0,
\forall k} g(\{\lambda_k\})$. Similar to Problem \ref{prob:LT LT
MAC}, it can be verified that the duality gap is zero for the convex
optimization problem addressed here; and thus solving its dual
problem is equivalent to solving the original problem.

Consider first the problem for obtaining $g(\{\lambda_k\})$ with
some given $\lambda_k$'s. Similar to Problem \ref{prob:LT LT MAC},
this dual function can be decomposed into individual
sub-dual-functions, each for one fading state, i.e.,
\begin{eqnarray}\label{eq:dual function rewrite LT ST}
g(\{\lambda_k\})= \mathbb{E}\left[g'(\mv{\alpha})\right] + \sum_{k=1}^K\lambda_kP^{\rm LT}_k
\end{eqnarray}
where
\begin{align}\label{eq:dual function per state LT ST}
g'(\mv{\alpha})=\max_{\{p_k(\mv{\alpha})\}\in\mathcal{A}(\mv{\alpha})}\log(1+\sum_{k=1}^Kh_kp_k(\mv{\alpha}))-
\sum_{k=1}^K\lambda_kp_k(\mv{\alpha})
\end{align}
with $\mathcal{A}(\mv{\alpha})$ denoting the subset of $\mathcal{A}$
corresponding to the fading state with channel realization
$\mv{\alpha}$. After dropping the $\mv{\alpha}$ in the corresponding
maximization problem in (\ref{eq:dual function per state LT ST}) for
a particular fading state, we can express this problem as
\begin{problem}\label{prob:LT ST MAC per state}
\begin{eqnarray}
\mathop{\mathtt{Max.}}_{\{p_k\}} &&
\log\left(1+\sum_{k=1}^Kh_kp_k\right)-  \sum_{k=1}^K\lambda_kp_k
\\ \mathtt {s.t.} && \label{eq:constraint 1 LT ST}
\sum_{k=1}^Kg_{km}p_k\leq \Gamma_m^{\rm ST},\ \forall m
\\ && p_k\geq 0, \ \forall k. \label{eq:constraint 2 LT ST}
\end{eqnarray}
\end{problem}
The above problem is convex, but in general does not have a
closed-form solution. Nevertheless, it can be efficiently solved by
standard convex optimization techniques, e.g., the interior point
method \cite{Boydbook}, or alternatively,  via solving its dual
problem; and for brevity, we omit the details here. After solving
Problem \ref{prob:LT ST MAC per state} for all the fading states, we
can obtain the dual function $g(\{\lambda_k\})$. Next, the
minimization of $g(\{\lambda_k\})$ over $\lambda_k$'s can be
resolved via the ellipsoid method, similarly like that for Problem
\ref{prob:LT LT MAC}.

For this case, we next focus on studying the conditions under which
D-TDMA is optimal across the fading states. This can be done by
investigating the KKT optimality conditions for Problem \ref{prob:LT
ST MAC per state}. First, we introduce nonnegative dual variables
$\mu_m$, $m=1,\ldots,M$, and $\delta_k, k=1,\ldots, K$, for their
associated constraints in (\ref{eq:constraint 1 LT ST}) and
(\ref{eq:constraint 2 LT ST}), respectively. The KKT conditions for
the optimal primal and dual solutions of this problem, denoted as
$\{p_k^*\}$, $\{\mu_m^*\}$, and $\{\delta_k^*\}$, can then be
expressed as
\begin{align}
\frac{h_k}{1+\sum_{l=1}^Kh_lp_l^*}-\lambda_k-\sum_{m=1}^M\mu_m^*g_{km}+\delta_k^*&=0,
\forall k \label{eq:KKT1 LT ST} \\
\mu_m^*\left(\sum_{k=1}^Kg_{km}p_k^* -\Gamma_m^{\rm ST}\right)&=0,
\forall m \label{eq:KKT2 LT ST} \\
\delta_k^*p_k^*&=0,  \forall k \label{eq:KKT3 LT ST}
\\ \sum_{k=1}^K g_{km}p_k^* &\leq \Gamma_m^{\rm ST}, \forall m
\label{eq:KKT4 LT ST}
\end{align}
with $p_k^*\geq 0, \forall k$, $\delta_k^*\geq 0, \forall k$, and
$\mu_m^*\geq 0, \forall m$. Notice that in this case $\mu_m$'s are
local variables for each fading state instead of being fixed as in
(\ref{eq:KKT1 LT LT}) for Problem \ref{prob:LT LT MAC per state}.
From these KKT conditions, the following lemma can then be obtained:
\begin{lemma}\label{lemma:TDMA opt LT ST}
The optimal solution of Problem \ref{prob:LT ST MAC per state} has
at most $M+1$ secondary users that transmit with strictly positive
power levels.
\end{lemma}
\begin{proof}
Please refer to Appendix \ref{appendix:proof TDMA opt LT ST}.
\end{proof}

Lemma \ref{lemma:TDMA opt LT ST} suggests that the optimal number of
SUs that can transmit at each fading state may depend on the number
of PRs or interference-power constraints. For small values of $M$,
e.g., $M=1$ corresponding to a single PR, the number of active SUs
at each fading state can be at most two, suggesting that D-TDMA may
be very close to being optimal in this case.

In the theorem below, we present the general conditions, for any $K$
and $M$, under which D-TDMA is both necessary and sufficient to be
optimal at a particular fading state. Again, without loss of
generality, here we use $\lambda_k$'s instead of their optimal dual
solutions obtained by the ellipsoid method.

\begin{theorem}\label{theorem:TDMA opt conditions LT ST}
D-TDMA is optimal at an arbitrary fading state for achieving the
ergodic sum capacity of the fading C-MAC under the LT-TPC jointly
with the ST-IPC if and only if there exists one user $i$ (the user
that transmits) that satisfies either one of the following two sets
of conditions. Let $j$ be any of the other users,
$j\in\{1,\ldots,K\}, j\neq i$; and
$m'=\arg\min_{m\in\{1,\ldots,M\}}\frac{\Gamma_m^{\rm ST}}{g_{im}}$.
\begin{itemize}
\item
$\frac{1}{\lambda_i}-\frac{1}{h_i}\leq\frac{\Gamma_{m'}^{\rm
ST}}{g_{im'}}$ and $\frac{h_i}{\lambda_i}\geq\frac{h_j}{\lambda_j},
\forall j\neq i$. In this case,
$p_i^*=\left(\frac{1}{\lambda_i}-\frac{1}{h_i}\right)^+$;

\item $\frac{1}{\lambda_i}-\frac{1}{h_i}>\frac{\Gamma_{m'}^{\rm
ST}}{g_{im'}}$ and
$\left(h_jg_{im'}-h_ig_{jm'}\right)\frac{g_{im'}}{g_{im'}+h_i\Gamma_{m'}^{\rm
ST}}\leq\left(\lambda_jg_{im'}-\lambda_ig_{jm'}\right), \forall
j\neq i$. In this case, $p_i^*=\frac{\Gamma_{m'}^{\rm
ST}}{g_{im'}}$.
\end{itemize}
\end{theorem}
\begin{proof}
Please refer to Appendix \ref{appendix:proof TDMA opt conditions LT
ST}.
\end{proof}

\begin{remark}
Notice that in Theorem \ref{theorem:TDMA opt conditions LT ST}, the
first set of conditions holds when the optimal transmit power of the
user with the largest $\frac{h_i}{\lambda_i}$ among all the users
satisfies the ST-IPC at all the PRs; the second set of conditions
holds when the first set fails to be true, and in this case any of
$K$ SUs can be the selected user for transmission provided that it
satisfies the given $K-1$ inequalities.
\end{remark}

\begin{remark}
In the special case where only the ST-IPC given by (\ref{eq:ST IPC
MAC}) is present or active in Problem \ref{prob:LT ST MAC}, all
$\lambda_k$'s in Theorem \ref{theorem:TDMA opt conditions LT ST} can
be taken as zeros. As a result, the first set of conditions can
never be true, while the second set of conditions are simplified as
$h_jg_{im'}-h_ig_{jm'}\leq 0, \forall j\neq i$, and the optimal
power of user $i$ that transmits is still
$p_i^*=\frac{\Gamma_{m'}^{\rm ST}}{g_{im'}}$. We thus have the
following corollary if it is further assumed that there is only a
single PR. For conciseness, the index $m$ for this PR is dropped
below.
\begin{corollary}\label{corol:LT ST}
In the case that only the ST-IPC given by (\ref{eq:ST IPC MAC}) is
present in Problem \ref{prob:LT ST MAC} and, furthermore, $M=1$,
D-TDMA is optimal; and the selected user $i$ for transmission
satisfies that $\frac{h_i}{g_{i}}\geq\frac{h_j}{g_{j}}, \forall
j\neq i$, with transmit power $p_i^*=\frac{\Gamma^{\rm ST}}{g_{i}}$.
\end{corollary}
\end{remark}

\subsection{Short-Term Transmit-Power and Long-Term
Interference-Power Constraints}

In the case of ST-TPC combined with LT-IPC, the ergodic sum capacity
is the optimal value of the following optimization problem:
\begin{problem}\label{prob:ST LT MAC}
\begin{eqnarray*}
\mathop{\mathtt{Max.}}_{\{p_k(\mv{\alpha})\}}&& \mathbb{E}\left[
\log\left(1+\sum_{k=1}^Kh_kp_k(\mv{\alpha})\right)\right]
\\ \mathtt {s.t.} && (\ref{eq:ST TPC MAC}), (\ref{eq:LT IPC
MAC}).
\end{eqnarray*}
\end{problem}
Again, we apply the Lagrange duality method for the above problem.
 Let $\mu_m$'s be the nonnegative dual variables associated
with the LT-IPC in (\ref{eq:LT IPC MAC}), $m=1,\ldots,M$. The
Lagrangian of Problem \ref{prob:ST LT MAC} can then be written as
\begin{align}\label{eq:Lagrangian ST LT}
\mathcal{L}(\{p_k(\mv{\alpha})\},\{\mu_m\})=\mathbb{E}\left[ \log\left(1+\sum_{k=1}^Kh_kp_k(\mv{\alpha})\right)\right]\nonumber \\  -
\sum_{m=1}^M\mu_m\left\{\mathbb{E}\left[\sum_{k=1}^Kg_{km}p_k(\mv{\alpha})\right]- \Gamma^{\rm LT}_m\right\}.
\end{align}
Let $\mathcal{B}$ denote the set of $\{p_k(\mv{\alpha})\}$ specified
by the remaining ST-TPC in (\ref{eq:ST TPC MAC}). The Lagrange dual
function is expressed as
\begin{eqnarray}\label{eq:Lagrange dual LT ST}
g(\{\mu_m\})=\max_{\{p_k(\mv{\alpha})\}\in\mathcal{B}}\mathcal{L}(\{p_k(\mv{\alpha})\},\{\mu_m\}).
\end{eqnarray}
The dual problem is accordingly defined as $\min_{\mu_m\geq 0,
\forall m} g(\{\mu_m\})$. Similar to the previous two cases, this
dual function can be equivalently written as
\begin{eqnarray}\label{eq:dual function rewrite ST LT}
g(\{\mu_m\})= \mathbb{E}\left[g'(\mv{\alpha})\right] + \sum_{m=1}^m\mu_k\Gamma^{\rm LT}_m
\end{eqnarray}
where
\begin{align}\label{eq:dual function per state ST LT}
g'(\mv{\alpha})=&~\max_{\{p_k(\mv{\alpha})\}\in\mathcal{B}(\mv{\alpha})}\log\left(1+\sum_{k=1}^Kh_kp_k(\mv{\alpha})\right)\nonumber \\ &~ -
\sum_{m=1}^K\mu_m\sum_{k=1}^Kg_{km}p_k(\mv{\alpha})
\end{align}
with $\mathcal{B}(\mv{\alpha})$ denoting the subset of $\mathcal{B}$
corresponding to the fading state with channel realization
$\mv{\alpha}$. After dropping $\mv{\alpha}$ in the maximization
problem in (\ref{eq:dual function per state ST LT}), for each
particular fading state we can express this problem as
\begin{problem}\label{prob:ST LT MAC per state}
\begin{eqnarray}
\mathop{\mathtt{Max.}}_{\{p_k\}} &&
\log\left(1+\sum_{k=1}^Kh_kp_k\right)-
\sum_{m=1}^K\mu_m\sum_{k=1}^Kg_{km}p_k
\\ \mathtt {s.t.} && \label{eq:constraint 1 ST LT}
p_k\leq P_k^{\rm ST}, \ \forall k
\\ && p_k\geq 0, \ \forall k. \label{eq:constraint 2 ST LT}
\end{eqnarray}
\end{problem}
After solving Problem \ref{prob:ST LT MAC per state} for all the
fading states, we obtain the dual function $g(\{\mu_m\})$. The dual
problem that minimizes $g(\{\mu_m\})$ over $\mu_m$'s can then be
solved again via the ellipsoid method.

Next, we present the closed-form solution of Problem \ref{prob:ST LT
MAC per state} based on its KKT optimality conditions. Let
$\lambda_k$ and $\delta_k$, $k=1,\ldots,K$, be the dual variables
for the corresponding user power constraints in (\ref{eq:constraint
1 ST LT}) and (\ref{eq:constraint 2 ST LT}), respectively. The KKT
conditions for the optimal primal and dual solutions of this
problem, denoted as $\{p_k^*\}$, $\{\lambda_k^*\}$, and
$\{\delta_k^*\}$, can then be expressed as
\begin{align}
\frac{h_k}{1+\sum_{l=1}^Kh_lp_l^*}-\lambda_k^*-\sum_{m=1}^M\mu_m g_{km}+\delta_k^*&=0,
\forall k \label{eq:KKT1 ST LT} \\
\lambda_k^*\left(p_k^* -P_k^{\rm ST}\right)&=0,
\forall k \label{eq:KKT2 ST LT} \\
\delta_k^*p_k^*&=0, \forall k \label{eq:KKT3 ST LT}
\\ p_k^* &\leq P_k^{\rm ST}, \forall k
\label{eq:KKT4 ST LT}
\end{align}
with $p_k^*\geq 0$, $\lambda_k^*\geq 0$, and $\delta_k^*\geq 0,
\forall k$. From these KKT conditions, the following lemma can be
first obtained:
\begin{lemma}\label{lemma:opt 1 ST LT}
Let $i$ and $j$ be any two arbitrary users,
$i,j\in\{1,2,\ldots,K\}$, with $p_i^*>0$ and $p_j^*=0$ in the
optimal solution of Problem \ref{prob:ST LT MAC per state}. Then, it
must be true that $\frac{h_i}{\sum_{m=1}^M\mu_m g_{im}}\geq
\frac{h_j}{\sum_{m=1}^M\mu_mg_{jm}}$.
\end{lemma}
\begin{proof}
Please refer to Appendix \ref{appendix:proof lemma 1 ST LT}.
\end{proof}

Let $\pi$ be a permutation over $\{1,\ldots,K\}$ such that
$\frac{h_{\pi(i)}}{\sum_{m=1}^M\mu_m g_{\pi(i)m}}\geq
\frac{h_{\pi(j)}}{\sum_{m=1}^M\mu_mg_{\pi(j)m}}$ if $i<j,
i,j\in\{1,\ldots,K\}$. Supposing that there are $|\mathcal{I}|$
users that can transmit with $\mathcal{I}\subseteq\{1,\ldots,K\}$
denoting this set of users, from Lemma \ref{lemma:opt 1 ST LT} it is
easy to verify that
$\mathcal{I}=\{\pi(1),\ldots,\pi(|\mathcal{I}|)\}$. The following
lemma then provides the closed-form solution to Problem \ref{prob:ST
LT MAC per state}:
\begin{lemma}\label{lemma:opt 2 ST LT}
The optimal solution of Problem \ref{prob:ST LT MAC per state} is
\begin{eqnarray}\label{eq:opt power ST LT}
p_{\pi(a)}^*=\left\{\begin{array}{ll} P_{\pi(a)}^{\rm ST} &
a<|\mathcal{I}| \\
\min\bigg(P_{\pi(|\mathcal{I}|)}^{\rm ST}, \bigg(\frac{h_{\pi(|\mathcal{I})|}}{\sum_{m=1}^M\mu_mg_{\pi(|\mathcal{I}|)m}}-1 \nonumber &\\
~~~~~~~-\sum_{b=1}^{|\mathcal{I}|-1}h_{\pi(b)}P_{\pi(b)}^{\rm ST}\bigg)\frac{1}{h_{\pi(|\mathcal{I}|)}}\bigg) & a=|\mathcal{I}|
\\ 0 & a> |\mathcal{I}|
\end{array} \right.
\end{eqnarray}
where $|\mathcal{I}|$ is the largest value of $x$ such that
$\frac{h_{\pi(x)}}{\sum_{m=1}^M\mu_mg_{\pi(x)m}}> 1
+\sum_{b=1}^{x-1}h_{\pi(b)}P_{\pi(b)}^{\rm ST}$.
\end{lemma}
\begin{proof}
Please refer to Appendix \ref{appendix:proof lemma 2 ST LT}.
\end{proof}

From Lemma \ref{lemma:opt 2 ST LT}, it follows that in the case of
ST-TPC along with LT-IPC, for the active secondary users at one
fading state, there is at most one user that transmits with power
lower than its ST power constraint, while all the other active users
transmit with their maximum powers.

Furthermore, from Lemma \ref{lemma:opt 2 ST LT}, we can derive the
conditions for the optimality of D-TDMA  at any fading state, which
are stated in the following theorem. Again, without loss of
generality, we use $\mu_m$'s instead of their optimal dual solutions
for Problem \ref{prob:ST LT MAC} in expressing these conditions.
\begin{theorem}\label{theorem:TDMA opt conditions ST LT}
D-TDMA is optimal at an arbitrary fading state for achieving the
ergodic sum capacity of the fading C-MAC under the ST-TPC jointly
with the LT-IPC if and only if user $\pi(1)$ satisfies
\begin{equation}\label{eq:TDMA opt ST LT}
1+h_{\pi(1)}P_{\pi(1)}^{\rm ST}\geq
\frac{h_{\pi(2)}}{\sum_{m=1}^M\mu_m g_{\pi(2)m}}.
\end{equation}
User $\pi(1)$ is then selected for transmission and its optimal
transmit power is
\begin{equation}\label{eq:TDMA opt power ST LT}
p_{\pi(1)}^*=\min\left(P_{\pi(1)}^{\rm ST},
\left(\frac{1}{\sum_{m=1}^M\mu_mg_{\pi(1)m}}-\frac{1}{h_{\pi(1)}}\right)^+
\right).
\end{equation}
\end{theorem}
\begin{proof}
From Lemma \ref{lemma:opt 2 ST LT}, it follows that D-TDMA is
optimal, i.e., $|\mathcal{I}|\leq 1$, occurs if and only if
(\ref{eq:TDMA opt ST LT}) holds. Then, (\ref{eq:TDMA opt power ST
LT}) is obtained from Lemma \ref{lemma:opt 2 ST LT} by combining the
cases of $|\mathcal{I}|=0$ and $|\mathcal{I}|=1$ .
\end{proof}

\begin{remark}
In the case of the traditional fading SISO-MAC with the user ST-TPC
given in (\ref{eq:ST TPC MAC}), but without the LT-IPC given in
(\ref{eq:LT IPC MAC}), it can be easily verified that the ergodic
sum capacity is achieved when all users transmit with their maximum
available power values given by $P_k^{\rm ST}$'s at each fading
state. This is consistent with the results obtained in (\ref{eq:opt
power ST LT}) by having all $\mu_m$'s associated with the LT-IPC
take zero values. With zero $\mu_m$'s, it can be easily verified
that the condition given in Theorem \ref{theorem:TDMA opt conditions
ST LT} is never satisfied, and thus D-TDMA cannot be optimal in this
special case.
\end{remark}

\subsection{Short-Term Transmit-Power and Interference-Power Constraints}

The ergodic sum capacity under both the ST-TPC and ST-IPC can be
obtained by solving the following optimization problem:
\begin{problem}\label{prob:ST ST MAC}
\begin{eqnarray*}
\mathop{\mathtt{Max.}}_{\{p_k(\mv{\alpha})\}}&& \mathbb{E}\left[
\log\left(1+\sum_{k=1}^Kh_kp_k(\mv{\alpha})\right)\right]
\\ \mathtt {s.t.} && (\ref{eq:ST TPC MAC}), (\ref{eq:ST IPC
MAC}).
\end{eqnarray*}
\end{problem}
Notice that this case differs from all three previous cases in that
all of its power constraints are short-term constraints and thus
separable over fading states. Therefore, we can decompose the
original problem into individual subproblems each for one fading
state. For conciseness, we drop again the $\mv{\alpha}$ and express
the rate maximization problem at a particular fading state as
\begin{problem}\label{prob:ST ST MAC per state}
\begin{eqnarray}
\mathop{\mathtt{Max.}}_{\{p_k\}} &&
\log\left(1+\sum_{k=1}^Kh_kp_k\right)
\\ \mathtt {s.t.} && \label{eq:constraint 1 ST ST}
p_k\leq P_k^{\rm ST}, \ \forall k
\\ && \sum_{k=1}^Kg_{km}p_k\leq \Gamma^{\rm ST}_m, \ \forall m \label{eq:constraint 2 ST ST}
\\ && p_k\geq 0, \ \forall k. \label{eq:constraint 3 ST ST}
\end{eqnarray}
\end{problem}
The above problem is convex, but in general does not have a
closed-form solution. Similar to Problem \ref{prob:LT ST MAC per
state}, the interior point method \cite{Boydbook} or the Lagrange
duality method can be used to solve this problem and thus we omit
the details here.

For this case, we next present in the following theorem the
conditions for D-TDMA to be optimal at an arbitrary fading state:
\begin{theorem}\label{theorem:TDMA opt conditions ST ST}
D-TDMA is optimal at an arbitrary fading state for achieving the
ergodic sum capacity of the fading C-MAC under the ST-TPC jointly
with the ST-IPC if and only if there exists one user $i$ (the user
that transmits) that satisfies both of the following two conditions.
Let $j$ be any of the other users, $j\in\{1,\ldots,K\}, j\neq i$,
and $m'=\arg\min_{m\in\{1,\ldots,M\}}\frac{\Gamma_m^{\rm
ST}}{g_{im}}$.
\begin{itemize}
\item $\frac{\Gamma_i^{\rm ST}}{g_{im'}}\leq P_i^{\rm ST}$;
\item $\frac{h_i}{g_{im'}}\geq \frac{h_j}{g_{jm'}}, \forall j\neq
i$.
\end{itemize}
The optimal transmit power of user $i$ is $p_i^*=\frac{\Gamma_i^{\rm
ST}}{g_{im'}}$.
\end{theorem}
\begin{proof}
Please refer to Appendix \ref{appendix:proof TDMA opt conditions ST
ST}.
\end{proof}

\section{Ergodic Sum Capacity for Fading Cognitive BC}
\label{sec:BC}

From (\ref{eq:sum capacity BC new}), the ergodic sum capacities for
the SISO fading C-BC under different mixed TPC and IPC constraints
can be obtained as the optimal values of the following optimization
problems:
\begin{problem}\label{prob:BC}
\begin{align}
\mathop{\mathtt{Max.}}_{\{p_k(\mv{\beta})\}}&~ \mathbb{E}\left[ \log\left(1+\sum_{k=1}^Kh_kp_k(\mv{\beta})\right)\right] \nonumber
\\ \mathtt {s.t.} &~ (\ref{eq:LT TPC BC}), (\ref{eq:LT IPC
BC}) \ ({\rm Case \ I: \ LT-TPC \ and \ LT-IPC}) \nonumber \\ &~ {\rm or} ~ (\ref{eq:LT TPC BC}), (\ref{eq:ST IPC BC}) \ ({\rm Case \ II: LT-TPC
\  and \ ST-IPC}) \nonumber \\ &~ {\rm or} ~
 (\ref{eq:ST TPC BC}), (\ref{eq:LT IPC BC}) \ ({\rm Case \ III: \
ST-TPC \ and \ LT-IPC}) \nonumber
\\ &~ {\rm or} ~  (\ref{eq:ST TPC BC}), (\ref{eq:ST IPC BC}) \ ({\rm Case \ IV:
\ ST-TPC \ and \ ST-IPC}). \nonumber
\end{align}
\end{problem}

Notice that in (\ref{eq:LT TPC BC})-(\ref{eq:ST IPC BC}), the
transmit power of the secondary BS at a given fading state,
$q(\mv{\beta})$, needs to be replaced by the user sum-power in the
dual C-MAC, $\sum_{k=1}^Kp_k(\mv{\beta})$. Compared with the
problems addressed in Section \ref{sec:MAC} for the C-MAC, it is
easy to see that the corresponding problems in the C-BC case are
very similar, e.g., both have the same objective function, and
similar affine constraints in terms of $p_k(\mv{\alpha})$'s or
$p_k(\mv{\beta})$'s. Thus, we skip the details of derivations and
present the results directly in the following theorem:

\begin{theorem}\label{theorem:TDMA opt conditions BC}
In each of Cases I-IV, D-TDMA is optimal across all the fading
states for achieving the ergodic sum capacity of the fading C-BC. In
each case, the user $i$ with the largest $h_i$ among all the users
should be selected for transmission at a particular fading state.
The optimal rule for assigning the transmit power of the BS at each
fading state (for conciseness $\mv{\beta}$ is dropped in the
following expressions) in each case is given below. Let $j$ be any
of the users other than $i$, $j\in\{1,\ldots,K\}, j\neq i$;
$m'=\arg\min_{m\in\{1,\ldots,M\}}\frac{\Gamma_m^{\rm ST}}{f_m}$; and
$\lambda$ and $\mu_m$'s are the optimal dual variables associated
with the LT-TPC in (\ref{eq:LT TPC BC}) and the LT-IPC in
(\ref{eq:LT IPC BC}), respectively, if they appear in any of the
following cases.
\begin{itemize}
\item Case I:
\begin{equation}\label{eq:BC power opt 1}
q^*=\left(\frac{1}{\lambda+\sum_{m=1}^M\mu_mf_m}-\frac{1}{h_i}\right)^+;
\end{equation}

\item Case II:
\begin{equation}\label{eq:BC power opt 2}
q^*=\min\left(\frac{\Gamma_{m'}^{\rm ST}}{f_{m'}},
\left(\frac{1}{\lambda}-\frac{1}{h_i}\right)^+\right);
\end{equation}

\item Case III:
\begin{equation}\label{eq:BC power opt 3}
q^*=\min\left(Q^{\rm ST},
\left(\frac{1}{\sum_{m=1}^M\mu_mf_m}-\frac{1}{h_i}\right)^+\right);
\end{equation}

\item Case IV:
\begin{equation}
q^*=\min\left(Q^{\rm ST}, \frac{\Gamma_{m'}^{\rm
ST}}{f_{m'}}\right).
\end{equation}

\end{itemize}
\end{theorem}

\begin{remark}
In the case of the traditional fading SISO-BC without the LT- or
ST-IPC, by combining the results in \cite{Knopp95} for the fading
SISO-MAC and the MAC-BC duality results in \cite{Goldsmith04}, it
can be inferred that it is optimal to deploy D-TDMA by transmitting
to the user with the largest $h_i$ at each time in terms of
maximizing the ergodic sum capacity, regardless of the LT- or ST-TPC
at the BS. Theorem \ref{theorem:TDMA opt conditions BC} can thus be
considered as the extensions of such result to the SISO fading C-BC
under the additional LT- or ST-IPC. Also notice that the optimal
power allocation strategies in (\ref{eq:BC power opt 1})-(\ref{eq:BC
power opt 3}) resemble the well-known
 ``water-filling (WF)'' solutions for the single-user fading
channels \cite{Coverbook}, \cite{Goldsmith97}.
\end{remark}

\section{Numerical Examples}\label{sec:numerical results}

In this section, we present numerical results on the performances of
the proposed multiuser DRA schemes for some example fading CR
networks under different mixed transmit-power and interference-power
constraints, namely: Case I: LT-TPC with (w/) LT-IPC; Case II:
LT-TPC w/ ST-IPC; Case III: ST-TPC w/ LT-IPC; and Case IV: ST-TPC w/
ST-IPC. For simplicity, we consider symmetric multiuser channels
where all channel complex coefficients are independent CSCG random
variables distributed as $\mathcal{CN}(0,1)$. In total, $10,000$
randomly generated channel power gain vectors for $\mv{\alpha}$ or
$\mv{\beta}$ are used to approximate the actual ergodic sum-rate of
the secondary network in each simulation result. Furthermore, we
assume that the TPC (LT or ST) values are identical for all SUs, and
the IPC (LT or ST) values are identically equal to one, the same as
the additive Gaussian noise variance, at all PRs. For convenience,
we use $P$ to stand for all $P^{\rm ST}_k$'s and $P^{\rm LT}_k$'s,
$Q$ for both $Q^{\rm ST}$ and $Q^{\rm LT}$, and $\Gamma$ for all
$\Gamma^{\rm ST}_m$'s and $\Gamma^{\rm LT}_m$'s. The simulation
results are presented in the following subsections.

\subsection{Effects of LT/ST TPC/IPC on Ergodic Sum Capacity}

First, we compare the achievable ergodic sum capacities for the
fading CR network under four different cases of mixed TPC and IPC.
Fig. \ref{fig:MAC} shows the results for the fading C-MAC with $K=2$
and $M=1$, and Fig. \ref{fig:BC} for the fading C-BC with $K=5$ and
$M=2$.

For the C-MAC case, it is observed in Fig. \ref{fig:MAC} that the
ergodic sum capacity $C_{\rm MAC}$ in Case I is always the largest
while that in Case IV is the smallest for any given SU transmit
power constraint $P$. This is as expected since both the ST-TPC and
ST-IPC are less favorable from the SU's perspective as compared to
their LT counterparts: The former one imposes more stringent power
constraints than the latter one over the DRA in the SU network. It
is also observed that as $P$ increases, eventually $C_{\rm MAC}$
becomes saturated as the IPC (LT or ST) becomes more dominant than
the TPC. On the other hand, for small values of $P$ where the TPC is
more dominant than the IPC, it is observed that the LT-TPC (where
D-TDMA is optimal in Case I and close to being optimal in Case II)
leads to a capacity gain over the ST-TPC (where D-TDMA is
non-optimal in Case III or IV) due to the well-known {\it multiuser
diversity} effect exploited by D-TDMA \cite{Tse02}. Furthermore,
$C_{\rm MAC}$ in Case II is observed to be initially larger than
that in Case III for small values of $P$, but becomes equal to and
eventually smaller than that in Case III as $P$ increases. This is
due to the facts that for small values of $P$, TPC dominates IPC and
furthermore LT-TPC is more flexible over ST-TPC; while for large
values of $P$, IPC becomes more dominant over TPC and LT-IPC is more
flexible over ST-IPC.

For the C-BC case, similar results like those in the C-MAC are
observed. However, there exists one quite different phenomenon for
the C-BC. As the secondary BS transmit power $Q$ becomes large, the
achievable ergodic sum capacity $C_{\rm BC}$ shown in Fig.
\ref{fig:BC} under the LT-IPC is much larger than that under the
ST-IPC, regardless of the LT- or ST-TPC, as compared with $C_{\rm
MAC}$ shown in Fig. \ref{fig:MAC}. This is due to the fact that for
the C-BC with $M=2$ and a single BS transmitter, the ST-IPC can
limit the transmit power of the secondary BS more stringently than
the case of C-MAC shown in Fig. \ref{fig:MAC}, where there are two
SU transmitters but only a single PR. Since it is not always the
case that both channels from the two SUs to the PR have very large
gains at a given time, in the C-MAC case the SU with the smallest
instantaneous channel gain to the PR can be selected for
transmission, i.e., there exists an interesting {\it new form of
multiuser diversity effect} in the fading C-MAC. In contrast, for
the C-BC, the BS is likely to transmit with large power only if both
channel gains from the BS to the two PRs are reasonably low.

\begin{figure}
\begin{center}
\scalebox{0.5}{\includegraphics{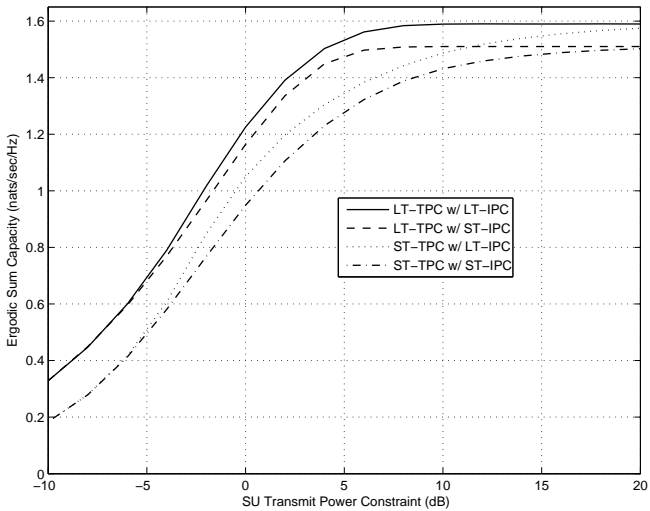}}
\end{center}
\caption{Comparison of the ergodic sum capacity under different combinations of TPC and IPC for the fading C-MAC with $K=2$,
$M=1$.}\label{fig:MAC}
\end{figure}

\begin{figure}
\begin{center}
\scalebox{0.5}{\includegraphics{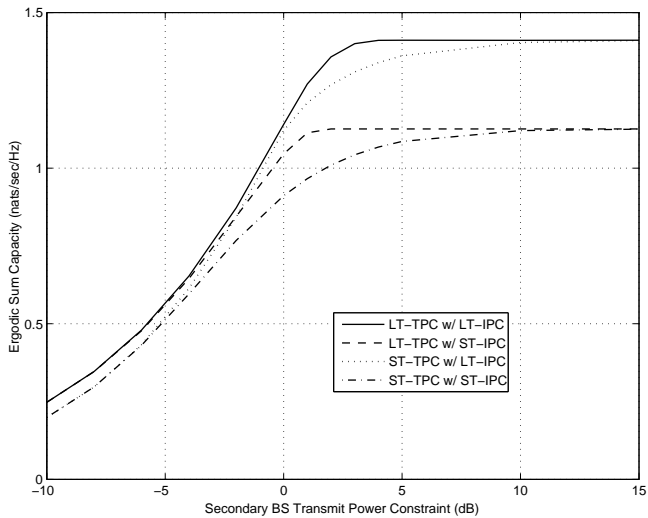}}
\end{center}
\caption{Comparison of the ergodic sum capacity under different combinations of TPC and IPC for the fading C-BC with $K=5$,
$M=2$.}\label{fig:BC}
\end{figure}

\subsection{Fading C-MAC With (w/) vs. Without (w/o) TDMA
Constraint}

Next, we consider the fading C-MAC and examine the effect of the
TDMA constraint on its achievable ergodic sum capacity. Notice that
for the fading C-BC, it has been shown in Theorem \ref{theorem:TDMA
opt conditions BC} that D-TDMA is optimal for all cases of mixed TPC
and IPC; and for the fading C-MAC, it has also been shown in Theorem
\ref{theorem:TDMA opt LT LT} that D-TDMA is optimal in Case I.
Therefore, in this subsection, we only consider the fading C-MAC in
Cases II, III, and IV. We compare the ergodic sum capacity $C_{\rm
MAC}$ achievable in each of these cases via the optimal DRA rule
proposed in this paper w/o the TDMA constraint against that with an
explicit TDMA constraint, i.e., at most one SU is selected for
transmission at any time. However, for the cases with the explicit
TDMA constraint, we still allow DRA over the SU network to optimally
select the SU (i.e., using D-TDMA) and set its power level for
transmission at each fading state, so as to maximize the long-term
average sum-rate. For conciseness, we discuss the optimal DRA
schemes for the fading C-MAC under the explicit TDMA constraint in
Appendix \ref{appendix:TDMA}.

In Figs. \ref{fig:MAC TDMA 2} and \ref{fig:MAC TDMA}, we compare the
achievable $C_{\rm MAC}$ w/ vs. w/o the TDMA constraint for Cases
II-IV with $K=2$, $M=1$, and $K=4, M=2$, respectively. It is
observed in both figures that the achievable $C_{\rm MAC}$ in each
case of mixed TPC and IPC is larger without the TDMA constraint.
This is as expected since TDMA is an additional constraint that
limits the flexibility of DRA in the SU network.

In Fig.  \ref{fig:MAC TDMA 2}, it is observed that the gap between
the achievable $C_{\rm MAC}$'s w/ and w/o the TDMA constraint in
each of Cases II-IV diminishes as the SU transmit power constraint
$P$ becomes sufficiently large. This phenomenon can be explained as
follows. First, note that as $P$ increases, eventually the TPC will
become inactive and the IPC becomes the only active power constraint
in each case. As a result, Case II and Case IV only have the (same)
ST-IPC and Case III only has the LT-IPC as active constraints. Thus,
the observed phenomenon is justified since D-TDMA has been shown to
be optimal for the above two cases, according to Corollary
\ref{corol:LT ST} (notice that $M=1$ for Fig. \ref{fig:MAC TDMA 2})
and Theorem \ref{theorem:TDMA opt LT LT} (with all $\lambda_k$'s
taking a zero value), respectively. However, in Fig. \ref{fig:MAC
TDMA} with $M>1$, only Case III has the same converged $C_{\rm MAC}$
w/ and w/o the TDMA constraint as $P$ becomes large, according to
Theorem \ref{theorem:TDMA opt LT LT}. In general, the capacity gap
between cases w/ and w/o the TDMA constraint becomes larger as $K$
or $M$ increases, as observed by comparing Figs. \ref{fig:MAC TDMA
2} and \ref{fig:MAC TDMA}. For example, for Case II, in Fig.
\ref{fig:MAC TDMA 2} with $M=1$, the capacity gap is negligible for
all values of $P$, which is consistent with Lemma \ref{lemma:TDMA
opt LT ST}; but it becomes notably large in Fig. \ref{fig:MAC TDMA}
with $M=2$.

\begin{figure}
\begin{center}
\scalebox{0.5}{\includegraphics{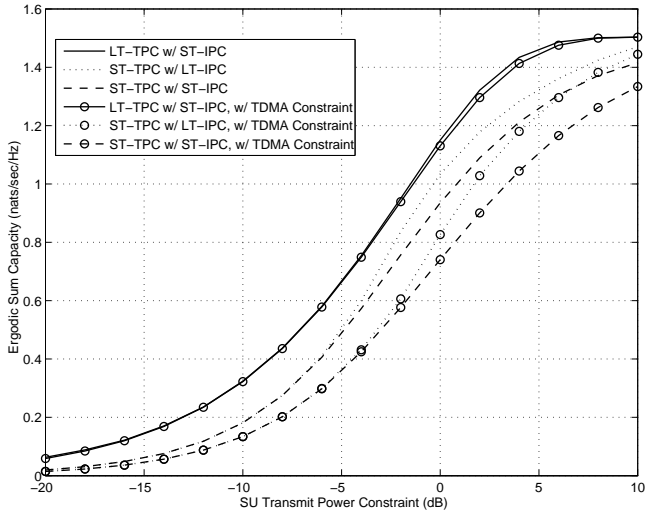}}
\end{center}
\caption{Comparison of the ergodic sum capacity w/ vs. w/o the TDMA constraint for the fading C-MAC with $K=2$, $M=1$.}\label{fig:MAC TDMA 2}
\end{figure}

\begin{figure}
\begin{center}
\scalebox{0.5}{\includegraphics{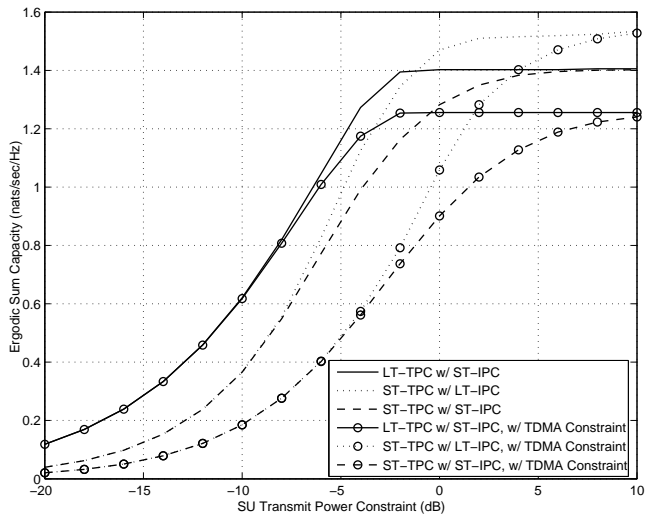}}
\end{center}
\caption{Comparison of the ergodic sum capacity w/ vs. w/o the TDMA constraint for the fading C-MAC with $K=4$, $M=2$.}\label{fig:MAC TDMA}
\end{figure}

\subsection{Dynamic vs. Fixed Resource Allocation}

At last, we compare the ergodic sum capacity achievable with the
optimal DRA against the achievable average sum-rate of users via
some heuristic fixed resource allocation (FRA) schemes for the same
fading CR network. For DRA, we select the most flexible power
allocation scheme for the SU network under the LT-TPC and the LT-IPC
(i.e., Case I), which is D-TDMA based and gives the largest $C_{\rm
MAC}$ and $C_{\rm BC}$ among all cases of mixed power constraints
under the same power-constraint values $P$ ($Q$) and $\Gamma$ for
the fading C-MAC (C-BC). For FRA, we also consider TDMA, which uses
the simple ``round-robin'' user scheduling rule, under the ST-TPC
and the ST-IPC. More specifically, for the fading C-MAC, at each
time the SU, say user $i$, which is scheduled for transmission, will
transmit a power equal to $\min(P,\frac{\Gamma}{\max_{m}g_{im}})$,
while for the fading C-BC, the BS transmits with the power equal to
$\min(Q,\frac{\Gamma}{\max_{m}f_{m}})$. Notice that the considered
FRA can be much more easily implemented as compared to the proposed
optimal DRA. Therefore, we need to examine the capacity gains by the
optimal DRA over the FRA.

In Fig. \ref{fig:MAC comp}, capacity comparisons between DRA and FRA
are shown for the fading C-MAC with $K=2$ or $4$, and $M=2$. Notice
that for the DRA case we have normalized the SU LT-TPC for $K=4$ by
a factor of $2$ such that the sum of user transmit power constraints
for both $K=2$ and $K=4$ are identical. Furthermore, for fair
comparison between DRA and FRA, the SU ST-TPC values in the FRA case
are $4$ and $2$ times the LT-TPC value in the DRA for $K=4$ and
$K=2$, respectively. It is observed that DRA achieves substantial
throughput gains over FRA for both $K=2$ and $K=4$. Notice that for
FRA, it can be easily shown that with the user power normalization,
the average sum-rate is statistically independent of $K$.
Furthermore, multiuser diversity gains in the achievable ergodic
sum-rate for the DRA are also observed by comparing $K=4$ against
$K=2$, given the same sum of user power constraints.

In Fig. \ref{fig:BC comp}, we show the capacity comparisons between
the fading C-BC with DRA and that with FRA, for a fixed secondary BS
transmit power constraint $Q=3$dB, $M=1$ or $4$, and different
values of $K$. Since there is only one transmitter at the BS for the
C-BC, there is no user power normalization required as in the C-MAC
case. The capacity gains by DRA over FRA are observed to become more
significant for both $M=1$ and $M=4$ cases, as $K$ increases, due to
the multiuser diversity effect. As an example, at $K=20$, the
capacities with DRA are $2.75$ and $3.83$ times of that with FRA,
for $M=1$ and $M=4$, respectively. This suggests that in contrast to
the conventional fading BC without any IPC, the multiuser diversity
gains obtained by the optimal DRA become more crucial to the fading
C-BC as the number of PRs, $M$, becomes larger.

\begin{figure}
\begin{center}
\scalebox{0.5}{\includegraphics{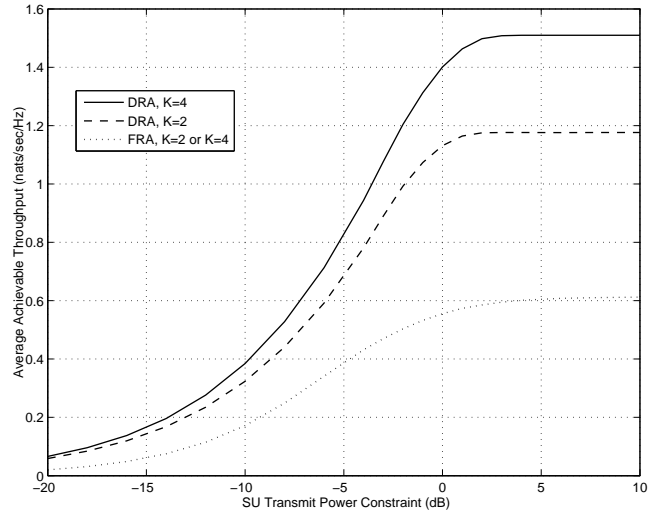}}
\end{center}
\caption{Comparison of the average achievable throughput with DRA
vs. with FRA for the fading C-MAC with $K=2$ or $4$,
$M=2$.}\label{fig:MAC comp}
\end{figure}

\begin{figure}
\begin{center}
\scalebox{0.5}{\includegraphics{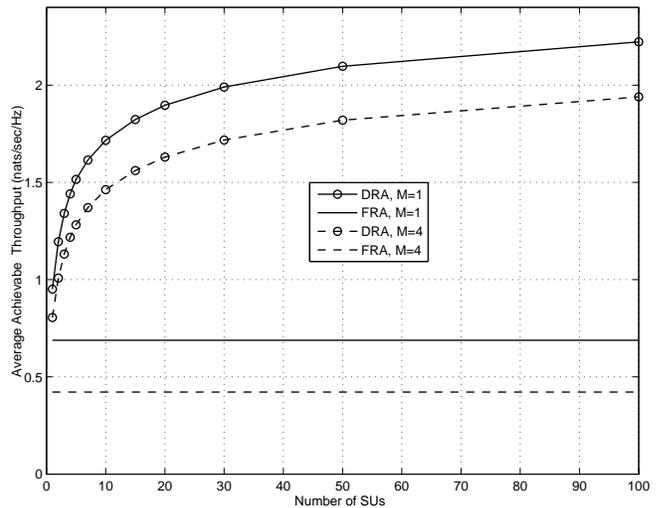}}
\end{center}
\caption{Comparison of the average achievable throughput with DRA
vs. with FRA for the fading C-BC with $M=1$ or $4$, and $Q^{\rm
LT}=Q^{\rm ST}=3$dB.} \label{fig:BC comp}
\end{figure}

\section{Concluding Remarks}\label{sec:conclusions}

In this paper, we have studied the information-theoretic limits of
the CR network under wireless spectrum sharing with an existing
primary radio network. By applying the interference-power constraint
as a practical means to protect each primary link, we characterize
the achievable ergodic sum capacity of the fading C-MAC and C-BC
under different mixed LT-/ST-TPC and LT-/ST-IPC. Optimal DRA schemes
for both cases w/ and w/o a TDMA constraint are presented.
Interestingly, except the cases where the optimality of D-TDMA can
be analytically proved, it is verified by simulation that there are
also many circumstances where D-TDMA with the optimal user
scheduling and power control performs very closely to the optimal
non-TDMA-based schemes in the fading C-MAC. Furthermore, an
interesting new form of multiuser diversity is observed for the
fading C-MAC by exploiting the additional CSI of channels between
secondary transmitters and primary receivers, which differs from
that in the conventional fading MAC by exploiting only the CSI of
channels between secondary users and BS.

Finally, it is worth pointing out that with the techniques
introduced in this paper, it is possible to derive the optimal
resource allocation for the more general cases where all LT/ST TPC
and IPC are present, and/or secondary users have different
priorities for rate allocation (i.e., characterization of the
capacity region instead of the sum capacity). Moreover, the results
in this paper are also applicable to the general channel models
consisting of parallel Gaussian channels over which the average and
instantaneous (transmit or interference) power constraints can be
applied, e.g., the frequency-selective fading broadband channel
which is decomposable into parallel narrow-band channels at each
fading state via the well-known
orthogonal-frequency-division-multiplexing (OFDM)
modulation/demodulation.

\appendices

\section{Proof of Lemma \ref{lemma:TDMA opt LT LT}} \label{appendix:proof TDMA opt LT LT}

Suppose that there are two arbitrary users $i$ and $j$ with
$p_i^*>0$ and $p_j^*>0$. From (\ref{eq:KKT2 LT LT}), it follows that
$\delta_i^*=0$ and $\delta_j^*=0$. Applying this fact to
(\ref{eq:KKT1 LT LT}), the following equality must hold:
\begin{equation}
\frac{h_i}{\lambda_i+\sum_{m=1}^M\mu_mg_{im}}=\frac{h_j}{\lambda_j+\sum_{m=1}^M\mu_mg_{jm}}.
\end{equation}
Since $h_i$ and $g_{im}$'s are independent of $h_j$ and $g_{jm}$'s,
and furthermore $\lambda_i$, $\lambda_j$, and $\mu_m$'s are all
constants in Problem \ref{prob:LT LT MAC per state}, it can be
inferred that the above equality is satisfied with a zero
probability. Thus, it is concluded that there is at most one user
with a strictly positive power value.

\section{Proof of Lemma \ref{lemma:opt power user LT LT}}
\label{appendix:proof opt user power LT LT}

Let user $i$ be the user that can transmit, i.e, $p_i^*> 0$, while
for the other users $j\neq i$, $p_j^*=0$. Problem \ref{prob:LT LT
MAC per state} then becomes the maximization of $\log(1+h_ip_i)-
\lambda_ip_i -\sum_{m=1}^M\mu_mg_{im}p_i$ subject to $p_i\geq 0$,
for which $p_i^*$ given in (\ref{eq:opt power LT LT}) can be easily
shown to be the optimal solution. Next, we need to show that for the
selected user $i$ for transmission, if $p_i^*>0$, it must satisfy
(\ref{eq:user ordering LT LT}). Since $p_i^*>0$, from (\ref{eq:KKT2
LT LT}) it follows that $\delta_i^*=0$. Since $\delta_j^*\geq 0,
\forall j\neq i$, from (\ref{eq:KKT1 LT LT}), it follows that
\begin{eqnarray}
\frac{h_i}{1+h_ip_i^*}-\lambda_i-\sum_{m=1}^M\mu_mg_{im}&=&0
\\ \frac{h_j}{1+h_ip_i^*}-\lambda_j-\sum_{m=1}^M\mu_mg_{jm}&\leq&0,
\ \forall j\neq i
\end{eqnarray}
from which (\ref{eq:user ordering LT LT}) can be obtained.

\section{Proof of Lemma \ref{lemma:TDMA opt LT ST}} \label{appendix:proof TDMA opt LT ST}

Suppose that there are $|\mathcal{J}|$ users with $p_j^*>0$, where
$j\in\mathcal{J}$ and $\mathcal{J}\subseteq\{1,2,\ldots,K\}$. Then
from (\ref{eq:KKT3 LT ST}), it follows that $\delta_j^*=0$, if
$j\in\mathcal{J}$. Let $c^*=1+\sum_{l=1}^Kh_lp_l^*$. From
(\ref{eq:KKT1 LT ST}), the following equalities must hold:
\begin{equation}
\frac{h_j}{c^*}-\lambda_j-\sum_{m=1}^M\mu_m^*g_{jm}=0, \forall
j\in\mathcal{J}.
\end{equation}
Removing $c^*$ in the above equations yields
\begin{equation}\label{eq:LT ST equalities}
\frac{\lambda_i+\sum_{m=1}^M\mu_m^*g_{im}}{h_i}=\frac{\lambda_j+\sum_{m=1}^M\mu_m^*g_{jm}}{h_j}, \forall j\in\mathcal{J}, j\neq i
\end{equation}
where $i$ is an arbitrary user index in $\mathcal{J}$. Notice that
in (\ref{eq:LT ST equalities}) there are $M$ variables
$\mu_1^*$,\ldots, $\mu_M^*$, but $|\mathcal{J}|-1$ independent
equations (with probability one). Therefore, $M\geq|\mathcal{J}|-1$
must hold in order for the above equations to have at least one set
of solutions. It then concludes that $|\mathcal{J}|$ must be no
greater than $M+1$.

\section{Proof of Theorem \ref{theorem:TDMA opt conditions LT ST}} \label{appendix:proof TDMA opt conditions LT ST}

Suppose that user $i$ transmits with $p_i^*>0$, while for the other
users $j\in\{1,\ldots,K\}, j\neq i$, $p_j^*=0$. We will consider the
following two cases: i) All $\mu_m^*$'s are equal to zero; ii) There
is one and only one $\mu_{m}^*$, denoted as $\mu_{m'}^*$, which is
strictly positive. Notice that it is impossible for more than one
$\mu_m^*$'s to be strictly positive at the same time, which can be
shown as follows. For user $i$, from (\ref{eq:KKT2 LT ST}),
$\mu_{m'}^*>0$ suggests that $g_{im'}p_i^*=\Gamma_{m'}^{\rm ST}$.
Supposing that there is $\tilde{m}\neq m'$ such that
$\mu_{\tilde{m}}^*>0$ and thus
$g_{i\tilde{m}}p_i^*=\Gamma_{\tilde{m}}^{\rm ST}$, a contradiction
then occurs as $\frac{g_{im'}}{\Gamma_{m'}^{\rm
ST}}=\frac{g_{i\tilde{m}}}{\Gamma_{\tilde{m}}^{\rm ST}}$ holds with
a zero probability.

First, we will prove the ``only if'' part of Theorem
\ref{theorem:TDMA opt conditions LT ST}. Consider initially the case
where all $\mu_m^*$'s are equal to zero. Suppose that $p_i^*>0$,
from (\ref{eq:KKT3 LT ST}) it follows that $\delta_i^*=0$. Since
$\delta_j^*\geq 0, \forall j\neq i$, from (\ref{eq:KKT1 LT ST}) the
followings must be true:
\begin{eqnarray}
\frac{h_i}{1+h_ip_i^*}-\lambda_i&=&0 \label{eq:equality LT ST} \\
\frac{h_j}{1+h_ip_i^*}-\lambda_j&\leq& 0, \ \forall j\neq i.
\end{eqnarray}
Thus, user $i$ must satisfy
$\frac{h_i}{\lambda_i}\geq\frac{h_j}{\lambda_j}, \forall j\neq i$.
From (\ref{eq:equality LT ST}), it follows that
$p_i^*=\left(\frac{1}{\lambda_i}-\frac{1}{h_i}\right)^+$ in this
case. Also notice that from (\ref{eq:KKT4 LT ST})
$g_{im}p_i^*\leq\Gamma_m^{\rm ST}$ must hold for $\forall
m=1,\ldots,M$. Therefore, we conclude that
$p_i^*\leq\frac{\Gamma_{m'}^{\rm ST}}{g_{im'}}$, where
$m'=\arg\min_{m\in\{1,\ldots,M\}}\frac{\Gamma_m^{\rm ST}}{g_{im}}$,
and thus $\left(\frac{1}{\lambda_i}-\frac{1}{h_i}\right)^+\leq
\frac{\Gamma_{m'}^{\rm ST}}{g_{im'}}$. Therefore, the first set of
conditions in Theorem \ref{theorem:TDMA opt conditions LT ST} is
obtained.

In the second case where there is one and only one $\mu_{m'}^*>0$,
it follows from (\ref{eq:KKT2 LT ST}) that
$g_{im'}p_i^*=\Gamma_{m'}^{\rm ST}$. Since from (\ref{eq:KKT4 LT
ST}) we have $g_{im}p_i^*\leq\Gamma_m^{\rm ST}, \forall m\neq m'$,
it follows that $\frac{\Gamma_{m'}^{\rm
ST}}{g_{im'}}\leq\frac{\Gamma_m^{\rm ST}}{g_{im}}$, and thus, again,
$m'=\arg\min_{m\in\{1,\ldots,M\}}\frac{\Gamma_m^{\rm ST}}{g_{im}}$,
and $p_i^*=\frac{\Gamma_{m'}^{\rm ST}}{g_{im'}}$ in this case. From
(\ref{eq:KKT1 LT ST}), we have
\begin{equation}\label{eq:opt mu LT ST}
\mu_{m'}^*=\left(\frac{h_i}{1+h_ip_i^*}-\lambda_i\right)\frac{1}{g_{im'}}.
\end{equation}
Since $\mu_{m'}^*>0$, from (\ref{eq:opt mu LT ST}) it follows that
$\frac{1}{\lambda_i}-\frac{1}{h_i}>p_i^*=\frac{\Gamma_{m'}^{\rm
ST}}{g_{im'}}$. Furthermore, from (\ref{eq:KKT1 LT ST}), the
followings must be true:
\begin{eqnarray}
\frac{h_i}{1+h_ip_i^*}-\lambda_i-\mu_{m'}^*g_{im'}&=&0 \\
\frac{h_j}{1+h_ip_i^*}-\lambda_j- \mu_{m'}^*g_{jm'} &\leq& 0, \
\forall j\neq i.
\end{eqnarray}
Thus, we have
\begin{equation}
\frac{h_i}{\lambda_i+\mu_{m'}^*
g_{im'}}\geq\frac{h_j}{\lambda_j+\mu_{m'}^*g_{jm'}}, \ \forall j\neq
i.
\end{equation}
Substituting $\mu_{m'}^*$ in (\ref{eq:opt mu LT ST}) into the above
inequalities yields
\begin{equation}
\left(h_jg_{im'}-h_ig_{jm'}\right)\frac{g_{im'}}{g_{im'}+h_i\Gamma_{m'}^{\rm ST}}\leq\left(\lambda_jg_{im'}-\lambda_ig_{jm'}\right),
\end{equation}
$\forall j\neq i$. The second set of conditions in Theorem \ref{theorem:TDMA opt conditions LT ST} is thus obtained.

Next, the ``if'' part of Theorem \ref{theorem:TDMA opt conditions LT
ST} can be shown easily by the fact that for a strictly-convex
optimization problem, the KKT conditions are not only necessary but
also sufficient to be satisfied by the unique set of primal and dual
optimal solutions \cite{Boydbook}.

\section{Proof of Lemma \ref{lemma:opt 1 ST LT}} \label{appendix:proof lemma 1 ST LT}

Since $p_j^*=0$, $p_i^*>0$, from (\ref{eq:KKT2 ST LT}) and
(\ref{eq:KKT3 ST LT}) it follows that $\lambda_j^*=0$ and
$\delta_i^*=0$, respectively. Then, from (\ref{eq:KKT1 ST LT}) it
follows that
\begin{eqnarray}
\frac{h_i}{1+\sum_{l=1}^Kh_lp_l^*}-\sum_{m=1}^M\mu_m g_{im}&\geq&0
\\ \frac{h_j}{1+\sum_{l=1}^Kh_lp_l^*}-\sum_{m=1}^M\mu_m
g_{jm}&\leq&0.
\end{eqnarray}
From the above two inequalities, Lemma \ref{lemma:opt 1 ST LT} can
be easily shown.

\section{Proof of Lemma \ref{lemma:opt 2 ST LT}} \label{appendix:proof lemma 2 ST LT}

The following lemma is required for the proof of Lemma
\ref{lemma:opt 2 ST LT}:
\begin{lemma} \label{lemma:opt 3 ST LT}
The optimal solution of Problem \ref{prob:ST LT MAC per state} has
at most one user, indexed by $i$, which satisfies $0<p_i^*<P_i^{\rm
ST}$, where $i=\pi(|\mathcal{I}|)$; and the optimal sum-power of
transmitting users must satisfy
$\sum_{a=1}^{|\mathcal{I}|}h_{\pi(a)}p_{\pi(a)}^*=\frac{h_{\pi(|\mathcal{I}|)}}{\sum_{m=1}^M\mu_mg_{\pi(|\mathcal{I}|)m}}-1$.
\end{lemma}
\begin{proof}
Suppose that there are two users $i$ and $j$ with $0<p_i^*<P_i^{\rm
ST}$ and $0<p_j^*<P_i^{\rm ST}$. From (\ref{eq:KKT2 ST LT}) and
(\ref{eq:KKT3 ST LT}), it follows that $\lambda_i^*=\lambda_j^*=0$
and $\delta_i^*=\delta_j^*=0$, respectively. Using these facts, from
(\ref{eq:KKT1 ST LT}), it follows that the following two equalities
must hold at the same time:
\begin{eqnarray}
\frac{h_i}{1+\sum_{l=1}^Kh_lp_l^*}-\sum_{m=1}^M\mu_m g_{im}&=&0
\label{eq:sum power ST LT} \\
\frac{h_j}{1+\sum_{l=1}^Kh_lp_l^*}-\sum_{m=1}^M\mu_m g_{jm}&=&0.
\end{eqnarray}
Thus, we have
\begin{equation}
\frac{h_i}{\sum_{m=1}^M\mu_mg_{im}}=\frac{h_j}{\sum_{m=1}^M\mu_mg_{jm}}.
\end{equation}
Since $h_i$ and $g_{im}$'s are independent of $h_j$ and $g_{jm}$'s,
and $\mu_m$'s are constants, it is inferred that the above equality
is satisfied with a zero probability. Thus, we conclude that there
is at most one user $i$ with $0<p_i^*<P_i^{\rm ST}$. From
(\ref{eq:sum power ST LT}), we have
\begin{equation}\label{eq:equality ST LT}
\sum_{l=1}^Kh_lp_l^*=\sum_{a=1}^{|\mathcal{I}|}h_{\pi(a)}p_{\pi(a)}^*=\frac{h_i}{\sum_{m=1}^M\mu_mg_{im}}-1.
\end{equation}
Using (\ref{eq:KKT1 ST LT}) and (\ref{eq:equality ST LT}), it is
easy to see that for any user $k\in\mathcal{I}, k\neq i$ with
$p_k^*>0$, it must satisfy
\begin{equation}
\frac{h_k}{\sum_{m=1}^M\mu_m g_{km}}\geq
\frac{h_i}{\sum_{m=1}^M\mu_m g_{im}}.
\end{equation}
Thus, we conclude that $i=\pi(|\mathcal{I}|)$.
\end{proof}

Lemma \ref{lemma:opt 3 ST LT} suggests that only one of the
following two sets of solutions for $p_k^*, k\in\mathcal{I},$ can be
true, which are
\begin{itemize}
\item Case I: $p_{\pi(a)}^*=P_{\pi(a)}^{\rm ST},
a=1,\ldots,|\mathcal{I}|$;

\item Case II: $p_{\pi(a)}^*=P_{\pi(a)}^{\rm ST},
a=1,\ldots,|\mathcal{I}|-1$, and
$p_{\pi(|\mathcal{I}|)}^*=\left(\frac{h_{\pi(|\mathcal{I})|}}{\sum_{m=1}^M\mu_mg_{\pi(|\mathcal{I}|)m}}-1
-\sum_{b=1}^{|\mathcal{I}|-1}h_{\pi(b)}P_{\pi(b)}^{\rm
ST}\right)\frac{1}{h_{\pi(|\mathcal{I}|)}}$.
\end{itemize}
Since $p_{\pi(|\mathcal{I}|)}^*\leq P_{\pi(|\mathcal{I}|)}^{\rm
ST}$, it then follows that
\begin{align}
p_{\pi(|\mathcal{I}|)}^* =&~ \min\bigg(P_{\pi(|\mathcal{I}|)}^{\rm ST},
\bigg(\frac{h_{\pi(|\mathcal{I})|}}{\sum_{m=1}^M\mu_mg_{\pi(|\mathcal{I}|)m}}-1 \nonumber \\ &~~~~~~~~~~~~~~~~~~~~
-\sum_{b=1}^{|\mathcal{I}|-1}h_{\pi(b)}P_{\pi(b)}^{\rm ST}\bigg)\frac{1}{h_{\pi(|\mathcal{I}|)}}\bigg).
\end{align}

The remaining part to be shown for Lemma \ref{lemma:opt 2 ST LT} is
that the optimal number of active users $|\mathcal{I}|$ is the
largest value of $x$ such that
\begin{equation}\label{eq:inequality ST LT}
\frac{h_{\pi(x)}}{\sum_{m=1}^M\mu_mg_{\pi(x)m}}> 1
+\sum_{b=1}^{x-1}h_{\pi(b)}P_{\pi(b)}^{\rm ST}.
\end{equation}
First, we show that in both Case I and Case II, for any user
$\pi(a)\in\mathcal{I}, a=1,\ldots,|\mathcal{I}|$, the above
inequality holds. Since for (\ref{eq:inequality ST LT}), from Lemma
\ref{lemma:opt 1 ST LT} it follows that its left-hand side decreases
as $x$ increases, while its right-hand side increases with $x$, it
is sufficient to show that (\ref{eq:inequality ST LT}) holds for
$a=|\mathcal{I}|$. This is the case since from (\ref{eq:KKT1 ST LT})
with $\delta_{\pi(|\mathcal{I}|)}^*=0$ and
$\lambda_{\pi(|\mathcal{I}|)}^*\geq 0$, we have
\begin{eqnarray}
\frac{h_{\pi(|\mathcal{I}|)}}{\sum_{m=1}^M\mu_mg_{\pi(|\mathcal{I}|)m}}
&\geq& 1 +\sum_{b=1}^{|\mathcal{I}|}h_{\pi(b)}p_{\pi(b)}^* \\ &>&
1+\sum_{b=1}^{|\mathcal{I}|-1}h_{\pi(b)}P_{\pi(b)}^{\rm ST}.
\end{eqnarray}
Next, we show that for any user $\pi(j), j\in
\{|\mathcal{I}|+1,\ldots,K\}$, (\ref{eq:inequality ST LT}) does not
hold. Again, it is sufficient to consider user
$\pi(|\mathcal{I}|+1)$ since if it does not satisfy
(\ref{eq:inequality ST LT}), neither does any of the other users
$\pi(|\mathcal{I}|+2),\ldots,\pi(K)$. For user
$\pi(|\mathcal{I}|+1)$, from (\ref{eq:KKT1 ST LT}) with
$\delta_{\pi(|\mathcal{I}|+1)}^*\geq 0$ and
$\lambda_{\pi(|\mathcal{I}|+1)}^*= 0$, it follows that
\begin{eqnarray}
\frac{h_{\pi(|\mathcal{I}|+1)}}{\sum_{m=1}^M\mu_mg_{\pi(|\mathcal{I}|+1)m}}
&\leq& 1 +\sum_{b=1}^{|\mathcal{I}|}h_{\pi(b)}p_{\pi(b)}^* \\ &\leq&
1+\sum_{b=1}^{|\mathcal{I}|}h_{\pi(b)}P_{\pi(b)}^{\rm ST}.
\end{eqnarray}
Therefore, it is concluded that (\ref{eq:inequality ST LT}) can be
used to determine $|\mathcal{I}|$.

\section{Proof of Theorem \ref{theorem:TDMA opt conditions ST ST}} \label{appendix:proof TDMA opt conditions ST ST}
The proof of Theorem \ref{theorem:TDMA opt conditions ST ST} is also
based on the KKT optimality conditions for Problem \ref{prob:ST ST
MAC per state}. Let $\lambda_k^*$, $\mu_m^*$, and $\delta_k^*$,
$k=1,\ldots,K, m=1,\ldots,M$ be the optimal dual variables
associated with the constraints in (\ref{eq:constraint 1 ST ST}),
(\ref{eq:constraint 2 ST ST}), and (\ref{eq:constraint 3 ST ST}),
respectively. The KKT conditions can then be expressed as
\begin{align}
\frac{h_k}{1+\sum_{l=1}^Kh_lp_l^*}-\lambda_k^*-\sum_{m=1}^M\mu_m^*g_{km} +\delta_k^*&=0,
\forall k \label{eq:KKT1 ST ST} \\
\lambda_k^*\left(p_k^* -P_k^{\rm ST}\right)&=0,
\forall k \label{eq:KKT2 ST ST} \\
\mu_m^*\left(\sum_{k=1}^K g_{km}p_{km}^*-\Gamma_m^{\rm ST}\right)&=0,
\forall m \label{eq:KKT3 ST ST} \\
\delta_k^*p_k^*&=0,  \forall k \label{eq:KKT4 ST ST}
\\ p_k^* &\leq P_k^{\rm ST}, \forall k \label{eq:KKT5 ST ST}
\\ \sum_{k=1}^K g_{km}p_{km}^*&\leq\Gamma_m^{\rm
ST}, \forall m \label{eq:KKT6 ST ST}
\end{align}
with $p_k^*\geq 0, \lambda_k^*\geq 0, \mu_m^*\geq 0$, and
$\delta_k^*\geq 0, \forall k, m$. First, we will prove the ``only
if'' part of Theorem \ref{theorem:TDMA opt conditions ST ST}.
Suppose that user $i$ should transmit with $p_i^*>0$, while for the
other users $ j\in\{1,\ldots,K\}, j\neq i$, $p_j^*=0$. From
(\ref{eq:KKT2 ST ST}) and (\ref{eq:KKT4 ST ST}), it follows that
$\lambda_j^*=0, \forall j\neq i$ and $\delta_i^*=0$, respectively.

We will show that there is one and only one $\mu_{m}^*$, denoted as
$\mu_{m'}^*$, which is strictly positive. Notice that it is
impossible for more than one $\mu_m^*$'s to be strictly positive at
the same time. For user $i$, from (\ref{eq:KKT3 ST ST}),
$\mu_{m'}^*>0$ suggests that $g_{im'}p_i^*=\Gamma_{m'}^{\rm ST}$.
Supposing that there is $\tilde{m}\neq m'$ such that
$\mu_{\tilde{m}}^*>0$ and thus
$g_{i\tilde{m}}p_i^*=\Gamma_{\tilde{m}}^{\rm ST}$, a contradiction
then occurs as $\frac{g_{im'}}{\Gamma_{m'}^{\rm
ST}}=\frac{g_{i\tilde{m}}}{\Gamma_{\tilde{m}}^{\rm ST}}$ holds with
a zero probability. Second, we will show that it is also impossible
for all $\mu_m^*$'s to be zero. If this is the case, (\ref{eq:KKT1
ST ST}) for any user $j\neq i$, becomes
$\frac{h_j}{1+h_ip_i^*}+\delta_j^*=0$. This can be true only when
$h_j=0$, which occurs with a zero probability. Therefore, we
conclude that there is one and only one $\mu_{m'}^*>0$.

Since $g_{im'}p_i^*=\Gamma_{m'}^{\rm ST}$ and from (\ref{eq:KKT6 ST
ST}) we have $g_{im}p_i^*\leq\Gamma_m^{\rm ST}, \forall m\neq m'$,
it follows that $\frac{\Gamma_{m'}^{\rm
ST}}{g_{im'}}\leq\frac{\Gamma_m^{\rm ST}}{g_{im}}$, and thus
$m'=\arg\min_{m\in\{1,\ldots,M\}}\frac{\Gamma_m^{\rm ST}}{g_{im}}$
and $p_i^*=\frac{\Gamma_{m'}^{\rm ST}}{g_{im'}}$. Also notice from
(\ref{eq:KKT5 ST ST}) that in this case $\frac{\Gamma_{m'}^{\rm
ST}}{g_{im'}}\leq P_i^{\rm ST}$ must hold. At last, considering
(\ref{eq:KKT1 ST ST}) for user $i$ and any other user $j$, we have
\begin{eqnarray}
\frac{h_i}{1+h_ip_i^*}-\mu_{m'}g_{im'} &=& 0 \\
\frac{h_j}{1+h_ip_i^*}-\mu_{m'}g_{jm'} &\leq& 0.
\end{eqnarray}
Thus, we conclude that $\frac{h_i}{g_{im'}}\geq \frac{h_j}{g_{jm'}},
\forall j\neq i$, must hold.

Next, the ``if'' part of Theorem \ref{theorem:TDMA opt conditions ST
ST} follows due to the fact that for a strictly-convex optimization
problem, the KKT conditions are both necessary and sufficient for
the unique set of primal and dual optimal solutions \cite{Boydbook}.

\section{Ergodic Sum Capacity for Fading C-MAC under TDMA Constraint} \label{appendix:TDMA}

In this appendix, we formally derive the optimal rule of user
selection and power control to achieve the ergodic sum capacity for
the SISO fading C-MAC {\it under an explicit TDMA constraint}, in
addition to any combination of transmit-power and interference-power
constraints. The TDMA constraint implies that at each fading state
there is only one SU that can transmit. Let $\Pi(\mv{\alpha})$ be a
mapping function that gives the index of the SU selected for
transmission at a fading state with channel realization
$\mv{\alpha}$. Note that for this particular fading state,
$p_{\Pi(\mv{\alpha})}\geq 0$, while for the other SUs
$k\in\{1,\ldots,K\}$, $k\neq \Pi(\mv{\alpha})$, $p_k=0$. The ergodic
sum capacity of the fading C-MAC under TDMA constraint can be
obtained as
\begin{equation}\label{eq:sum capacity MAC TDMA}
C_{\rm MAC}^{\rm TDMA}=\max_{\Pi(\mv{\alpha})}\max_{\{p_k(\mv{\alpha})\}\in\mathcal{F}}\mathbb{E}\left[
\log\left(1+h_{\Pi(\mv{\alpha})}p_{\Pi(\mv{\alpha})}(\mv{\alpha})\right)\right]
\end{equation}
where $\mathcal{F}$ is specified by a particular combination of
power constraints described in (\ref{eq:LT TPC MAC})-(\ref{eq:ST IPC
MAC}). Clearly, for any given function $\Pi(\mv{\alpha})$, the
capacity maximization in (\ref{eq:sum capacity MAC TDMA}) over
$\mathcal{F}$ is a convex optimization problem. However, the
maximization over the function $\Pi(\mv{\alpha})$ may not be
necessarily convex, and thus standard convex optimization techniques
may not apply directly. Fortunately, it will be shown next that the
optimization problem in (\ref{eq:sum capacity MAC TDMA}) can be
efficiently solved for all considered cases of mixed LT-/ST-TPC and
LT-/ST-IPC.

\subsection{Long-Term Transmit-Power and Interference-Power Constraints}

From (\ref{eq:sum capacity MAC TDMA}), the ergodic sum capacity
under the TDMA constraint, as well as the LT-TPC in (\ref{eq:LT TPC
MAC}) and the LT-IPC in (\ref{eq:LT IPC MAC}) can be obtained by
solving the following optimization problem:
\begin{problem}\label{prob:LT LT MAC TDMA}
\begin{align}
\mathop{\mathtt{Max.}}_{\Pi(\mv{\alpha}),\{p_k(\mv{\alpha})\}}&~~ \mathbb{E}\left[
\log\left(1+h_{\Pi(\mv{\alpha})}p_{\Pi(\mv{\alpha})}(\mv{\alpha})\right)\right] \nonumber
\\ \mathtt {s.t.} &~~
\mathbb{E}\left[p_k(\mv{\alpha})\cdot\mv{1}(\Pi(\mv{\alpha})=k)\right]\leq P^{\rm LT}_k, \ \forall k \label{eq:LT TPC MAC TDMA} \\ &~~
\mathbb{E}\left[g_{\Pi(\mv{\alpha})m}p_{\Pi(\mv{\alpha})}(\mv{\alpha})\right]\leq \Gamma^{\rm LT}_m, \ \forall m \label{eq:LT IPC MAC TDMA}
\end{align}
where $\mv{1}(A)$ is the indicator function taking the values of 1
or 0 depending on the trueness or falseness of event $A$,
respectively.
\end{problem}

\begin{figure*}
\begin{align}\label{eq:Lagrangian LT LT TDMA}
\mathbb{E}\left[ \log\left(1+h_{\Pi(\mv{\alpha})}p_{\Pi(\mv{\alpha})}(\mv{\alpha})\right)\right] -\sum_{k=1}^K\lambda_k\left\{
\mathbb{E}\left[p_k(\mv{\alpha})\cdot\mv{1}(\Pi(\mv{\alpha})=k)\right]- P^{\rm LT}_k\right\}
-\sum_{m=1}^M\mu_m\left\{\mathbb{E}\left[g_{\Pi(\mv{\alpha})m}p_{\Pi(\mv{\alpha})}(\mv{\alpha})\right]- \Gamma^{\rm LT}_m\right\}
\end{align}
\end{figure*}

First, we write the Lagrangian of this problem,
$\mathcal{L}(\Pi(\mv{\alpha}),\{p_k(\mv{\alpha})\},\{\lambda_k\},\{\mu_m\})$,
as in (\ref{eq:Lagrangian LT LT TDMA}) (shown on the next page),
where $\lambda_k$ and $\mu_m$ are the nonnegative dual variables
associated with the corresponding constraints in (\ref{eq:LT TPC MAC
TDMA}) and (\ref{eq:LT IPC MAC TDMA}), respectively, $k=1,\ldots,K$,
$m=1,\ldots,M$. Then, the Lagrange dual function,
$g(\{\lambda_k\},\{\mu_m\})$,  is defined as
\begin{eqnarray}\label{eq:Lagrange dual LT LT TDMA}
\max_{\Pi(\mv{\alpha}),\{p_k(\mv{\alpha})\}}\mathcal{L}(\Pi(\mv{\alpha}),\{p_k(\mv{\alpha})\},\{\lambda_k\},\{\mu_m\}).
\end{eqnarray}
The dual problem is accordingly defined as
$\min_{\{\lambda_k\},\{\mu_m\}} g(\{\lambda_k\},\{\mu_m\})$. Since
the problem at hand may not be convex, the duality gap between the
optimal values of the original and the dual problems may not be
zero. However, it will be shown in the later part of this subsection
that the duality gap for Problem \ref{prob:LT LT MAC TDMA} is indeed
zero.

We consider only the maximization problem in (\ref{eq:Lagrange dual
LT LT TDMA}) for obtaining $g(\{\lambda_k\},\{\mu_m\})$ with some
given $\lambda_k$'s and $\mu_m$'s, while the minimization of
$g(\{\lambda_k\},\{\mu_m\})$ over $\lambda_k$'s and $\mu_m$'s can be
obtained by the ellipsoid method, since it is always a convex
optimization problem. For each fading state, the maximization
problem in (\ref{eq:Lagrange dual LT LT TDMA}) can be expressed as
(with $\mv{\alpha}$ dropped for brevity)
\begin{problem}\label{prob:LT LT MAC per state TDMA}
\begin{eqnarray}
\mathop{\mathtt{Max.}}_{\Pi, p_{\Pi}} &&
\log\left(1+h_{\Pi}p_{\Pi}\right)- \lambda_{\Pi}p_{\Pi}
-\sum_{m=1}^M\mu_mg_{\Pi m}p_{\Pi}
\\ \mathtt {s.t.} && p_{\Pi}\geq 0.
\end{eqnarray}
\end{problem}
For any given user $\Pi$, the optimal power solution for the above
problem can be obtained as
\begin{equation}\label{eq:opt power LT LT TDMA}
p_{\Pi}^*=\left( \frac{1}{\lambda_{\Pi}+\sum_{m=1}^M\mu_m g_{\Pi
m}}-\frac{1}{h_{\Pi}}\right)^+.
\end{equation}
Substituting this solution into the objective function of Problem
\ref{prob:LT LT MAC per state TDMA} yields
\begin{equation}
(\log(\frac{h_{\Pi}}{\lambda_{\Pi}+\sum_{m=1}^M\mu_m g_{\Pi m}}))^+ -(1-\frac{\lambda_{\Pi}+\sum_{m=1}^M\mu_m g_{\Pi m}}{h_{\Pi}})^+.
\end{equation}
It is easy to verify that the maximization of the above function
over $\Pi$ is attained with user $i$ that satisfies
\begin{equation}\label{eq:user ordering LT LT TDMA}
\frac{h_i}{\lambda_i+\sum_{m=1}^M\mu_m g_{im}}\geq\frac{h_j}{\lambda_j+\sum_{m=1}^M\mu_m g_{jm}}, \forall j\neq i.
\end{equation}
From (\ref{eq:opt power LT LT TDMA}) and (\ref{eq:user ordering LT
LT TDMA}), it follows that the same set of solutions for Problem
\ref{prob:LT LT MAC per state} without the TDMA constraint, which is
given in Lemma \ref{lemma:opt power user LT LT}, also holds for
Problem \ref{prob:LT LT MAC per state TDMA} with the TDMA
constraint. Note that the optimal solutions of Problem \ref{prob:LT
LT MAC} without the TDMA constraint are also TDMA-based, and thus
they are also feasible solutions to Problem \ref{prob:LT LT MAC
TDMA} with the TDMA constraint. Since these solutions have also been
shown in the above to be optimal for the dual problem of Problem
\ref{prob:LT LT MAC TDMA}, we conclude that the duality gap is zero
for Problem \ref{prob:LT LT MAC TDMA}; and both Problem \ref{prob:LT
LT MAC} and Problem \ref{prob:LT LT MAC TDMA} have the same set of
solutions.

\subsection{Long-Term Transmit-Power and Short-Term
Interference-Power Constraints}

The ergodic sum capacity under the TDMA constraint plus the LT-TPC
and the ST-IPC can be obtained as the optimal value of the following
problem:
\begin{problem}\label{prob:LT ST MAC TDMA}
\begin{eqnarray}
\mathop{\mathtt{Max.}}_{\Pi(\mv{\alpha}),\{p_k(\mv{\alpha})\}}&&
\mathbb{E}\left[
\log\left(1+h_{\Pi(\mv{\alpha})}p_{\Pi(\mv{\alpha})}(\mv{\alpha})\right)\right]
\nonumber
\\ \mathtt {s.t.} && (\ref{eq:LT TPC MAC TDMA}) \nonumber
\\ &&
g_{\Pi(\mv{\alpha})m}p_{\Pi(\mv{\alpha})}(\mv{\alpha})\leq
\Gamma^{\rm ST}_m, \ \forall \mv{\alpha}, m. \label{eq:ST IPC MAC
TDMA}
\end{eqnarray}
\end{problem}
Similarly as for Problem \ref{prob:LT LT MAC TDMA}, we apply the
Lagrange duality method for solving the above problem by introducing
the nonnegative dual variables $\lambda_k, k=1,\ldots,K$, associated
with the LT-TPC given in (\ref{eq:LT TPC MAC TDMA}). However, since
Problem \ref{prob:LT ST MAC TDMA} is not necessarily convex, the
duality gap for this problem may not be zero. Nevertheless, it can
be verified that Problem \ref{prob:LT ST MAC TDMA} satisfies the
so-called ``time-sharing'' conditions \cite{Yu06} and thus has a
zero duality gap. For brevity, we skip the details of derivations
here and present the optimal power-control policy in this case as
follows:
\begin{lemma}
In the optimal solution of Problem \ref{prob:LT ST MAC TDMA}, the
user $\Pi(\mv{\alpha})$ that transmits at a fading state with
channel realization $\mv{\alpha}$ maximizes the following expression
among all the users (with $\mv{\alpha}$ dropped for brevity):
\begin{equation}
\log\left(1+h_{\Pi}p^*_{\Pi}\right)-\lambda_{\Pi}p^*_{\Pi}
\end{equation}
where
\begin{equation}
p^*_{\Pi}=\min\left(\min_{m\in\{1,\ldots,M\}}\frac{\Gamma_m^{\rm ST}}{g_{\Pi m}},
\left(\frac{1}{\lambda_{\Pi}}-\frac{1}{h_{\Pi}}\right)^+\right)
\end{equation}
and $\lambda_k, k=1,\ldots,K$, are the optimal dual solutions
obtained by the ellipsoid method.
\end{lemma}

\subsection{Short-Term Transmit-Power and Long-Term
Interference-Power Constraints}

The ergodic sum capacity under the TDMA constraint, the ST-TPC, and
the LT-IPC can be obtained as the optimal value of the following
problem:
\begin{problem}\label{prob:ST LT MAC TDMA}
\begin{eqnarray}
\mathop{\mathtt{Max.}}_{\Pi(\mv{\alpha}),\{p_k(\mv{\alpha})\}}&&
\mathbb{E}\left[
\log\left(1+h_{\Pi(\mv{\alpha})}p_{\Pi(\mv{\alpha})}(\mv{\alpha})\right)\right]
\nonumber
\\ \mathtt {s.t.} && p_{\Pi(\mv{\alpha})}(\mv{\alpha})\leq P^{\rm ST}_{\Pi(\mv{\alpha})}, \ \forall \mv{\alpha}
\label{eq:ST TPC MAC TDMA} \\
&& (\ref{eq:LT IPC MAC TDMA}). \nonumber
\end{eqnarray}
\end{problem}
By introducing the nonnegative dual variables $\mu_m, m=1,\ldots,M$,
associated with the LT-IPC given in (\ref{eq:LT IPC MAC TDMA}),
Problem \ref{prob:ST LT MAC TDMA} can be solved similarly as for
Problem \ref{prob:LT ST MAC TDMA} by the Lagrange duality method.
For brevity, we present the optimal power-control policy in this
case directly as follows:
\begin{lemma}
In the optimal solution of Problem \ref{prob:ST LT MAC TDMA}, the
user $\Pi(\mv{\alpha})$ that transmits at a fading state with
channel realization $\mv{\alpha}$ maximizes the following expression
among all the users (with $\mv{\alpha}$ dropped for brevity):
\begin{equation}
\log\left(1+h_{\Pi}p^*_{\Pi}\right)-\sum_{m=1}^M\mu_mg_{\Pi m}p^*_{\Pi}
\end{equation}
where
\begin{equation}
p^*_{\Pi}=\min\left(P_{\Pi}^{\rm ST}, \left(\frac{1}{\sum_{m=1}^M\mu_mg_{\Pi m}}-\frac{1}{h_{\Pi}}\right)^+\right)
\end{equation}
and $\mu_m, m=1,\ldots,M$, are the optimal dual solutions obtained
by the ellipsoid method.
\end{lemma}

\subsection{Short-Term Transmit-Power and Interference-Power Constraints}

At last, the ergodic sum capacity under the TDMA constraint, the
ST-TPC, and the ST-IPC can be obtained as the optimal value of the
following problem:
\begin{problem}\label{prob:ST ST MAC TDMA}
\begin{eqnarray}
\mathop{\mathtt{Max.}}_{\Pi(\mv{\alpha}),\{p_k(\mv{\alpha})\}}&&
\mathbb{E}\left[
\log\left(1+h_{\Pi(\mv{\alpha})}p_{\Pi(\mv{\alpha})}(\mv{\alpha})\right)\right]
\nonumber
\\ \mathtt {s.t.} && (\ref{eq:ST TPC MAC TDMA}), (\ref{eq:ST IPC MAC
TDMA}). \nonumber
\end{eqnarray}
\end{problem}
In this case, all the constraints are separable over the fading
states and, thus, this problem is decomposable into independent
subproblems each for one fading state. For brevity, we present the
optimal power-control policy in this case directly as follows:
\begin{lemma}
In the optimal solution of Problem \ref{prob:ST ST MAC TDMA}, the
user $\Pi(\mv{\alpha})$ that transmits at a fading state with
channel realization $\mv{\alpha}$ maximizes the following expression
among all the users (with $\mv{\alpha}$ dropped for brevity):
\begin{equation}
p_{\Pi}^*h_{\Pi}
\end{equation}
where
\begin{equation}
p^*_{\Pi}=\min\left(P_{\Pi}^{\rm ST},
\min_{m\in\{1,\ldots,M\}}\frac{\Gamma_m^{\rm ST}}{g_{\Pi m}}\right).
\end{equation}

\end{lemma}

\end{document}